\documentclass[runningheads]{llncs}
\usepackage{amsmath,amssymb}
\usepackage{cite}
\usepackage[inline]{enumitem}
\usepackage[T1]{fontenc}
\usepackage[scaled=0.85]{beramono}
\usepackage{nicefrac}
\usepackage{xcolor}
\usepackage{tikz,pgfplots} \pgfplotsset{compat=1.18}
\usepackage{mathtools,calc}
\usetikzlibrary{graphs, quotes, backgrounds, fit}
\usepackage{wrapfig}
\usepackage{fontawesome}
\usepackage{pdflscape}
\usepackage{multirow}
\usepackage{afterpage}
\usepackage{orcidlink}
\hypersetup{hidelinks}

\newlength{\digitwd}
\AtBeginDocument{\setlength{\digitwd}{\widthof{$0$}}}
\newcommand{\digitx}{\mathmakebox[\digitwd][c]{\vphantom{0}\star}}

\makeatletter

\@addtoreset{theorem}{section}

\let\c@lemma\c@theorem

\let\c@proposition\c@theorem

\let\c@corollary\c@theorem

\let\c@definition\c@theorem

\let\c@remark\c@theorem

\let\c@example\c@theorem

\makeatother
\newcommand{\originalautomaton}{\mathcal{A}}
\newcommand{\isomorphautomaton}{\mathcal{H}}
\newcommand{\approximateautomaton}{\mathcal{T}}

\newcommand{\task}{\tau}

\newcommand{\weaklyhardconstraint}{\texttt{Const}}
\newcommand{\trim}[3]{\ensuremath{\texttt{\textbf{Trim}}\,(#1, #2, #3)}}
\newcommand{\anymiss}[2]{\ensuremath{\texttt{\textbf{AnyMiss}}\,(#1, #2)}}
\newcommand{\anyhit}[2]{\ensuremath{\texttt{\textbf{AnyHit}}\,(#1, #2)}}
\newcommand{\rowmiss}[2]{\ensuremath{\texttt{\textbf{RowMiss}}\,(#1, #2)}}
\newcommand{\rowhit}[2]{\ensuremath{\texttt{\textbf{RowHit}}\,(#1, #2)}}
\newcommand{\weaklyhardh}{h}
\newcommand{\weaklyhardm}{m}
\newcommand{\weaklyhardk}{k}
\newcommand{\alphabet}{\Sigma}
\newcommand{\alphabetletter}{\sigma}
\newcommand{\deadlinemiss}{0}
\newcommand{\deadlinehit}{1}
\newcommand{\traceitem}{a}

\newcommand{\adj}{\Pi}
\newcommand{\adjelement}{\pi}
\newcommand{\adjm}{\Pi_{\deadlinemiss}}
\newcommand{\adjh}{\Pi_{\deadlinehit}}
\newcommand{\cl}{\Phi}
\newcommand{\clAm}{\cl_{\deadlinemiss}}
\newcommand{\clAh}{\cl_{\deadlinehit}}

\newcommand{\approxim}{c}
\newcommand{\lang}{\mathcal{L}}

\newcommand{\wordlength}{\ell}
\newcommand{\langlength}[1]{#1^{=\wordlength}}

\newcommand{\cardinality}[1]{|#1|}

\newcommand{\comprate}{\approxim}

\newcommand{\aha}{\originalautomaton_{\weaklyhardm, \weaklyhardk}}
\newcommand{\ahav}{V}
\newcommand{\ahae}{E}
\newcommand{\caha}{\isomorphautomaton_{\weaklyhardm, \weaklyhardk}}
\newcommand{\cahav}{U}
\newcommand{\cahae}{F}
\newcommand{\compaha}{\approximateautomaton_{\weaklyhardm, \weaklyhardk, \comprate}}
\newcommand{\compahav}{U_{\comprate}}
\newcommand{\compahae}{F_{\comprate}}
\newcommand{\critahav}{U_{\comprate, \text{critical}}}
\newcommand{\difference}{d}
\newcommand{\critdahav}{U_{\comprate, \text{critical}, \difference}}
\newcommand{\e}{e}
\newcommand{\critdeahav}{U_{\comprate, \text{critical}, \difference, \e}}
\newcommand{\noncritahav}{U_{\comprate, \text{non-critical}}}

\newcommand{\ahanode}{v}
\newcommand{\ahainitnode}{v_{\text{init}}}
\newcommand{\ahainitnodearg}[1]{v_{\text{init}, #1}}
\newcommand{\cahanode}{u}
\newcommand{\cahainitnode}{u_{\text{init}}}
\newcommand{\cahaedge}{f}

\newcommand{\indexfunc}{g}
\newcommand{\mapfunc}{n}
\newcommand{\bijfunc}{f}
\newcommand{\indextuple}{p}

\newcommand{\natnumbers}{\mathbb{N}}
\newcommand{\realnumbers}{\mathbb{R}}
\newcommand{\realnumbersmatrix}[2]{\mathbb{R}^{#1 \times #2}}

\newcommand{\binaryrel}{R}
\newcommand{\binaryrelsign}{\lesssim}
\newcommand{\binaryrelsignlong}{\binaryrelsign_\binaryrel}

\newcommand{\abs}[1]{\left \lvert #1 \right \rvert}
\newcommand{\norm}[1]{\left \lVert #1 \right \rVert}
\newcommand{\onenorm}[1]{\left \lVert #1 \right \rVert_1}

\newcommand{\initstatevector}{p}
\newcommand{\accstatevector}{f}
\newcommand{\eigenvalue}{\lambda}
\newcommand{\eigenvectorx}{x}
\newcommand{\eigenvectory}{y}

\title{Over-approximation of weakly-hard constraints for control systems verification (Extended)}
\titlerunning{Over-approximation of weakly-hard constraints}
\author{Rieke de Maeyer, Holger Hermanns\,\orcidlink{0000-0002-2766-9615}, Martina Maggio\thanks{Corresponding author.}\,\orcidlink{0000-0002-1143-1127}}
\institute{Department of Computer Science, Saarland University\\Saarland Informatics Campus}

\begin{document}

\maketitle

\begin{abstract}
A hard real-time system cannot miss any deadline.
A weakly-hard real-time system, on the contrary, is designed to tolerate a specific number of deadline misses.
For instance, the \anymiss{2}{300} weakly-hard constraint stipulates that in every window of $300$ consecutive jobs, at most $2$ deadlines are missed.
The weakly-hard model is the state-of-the-art for industrial dependability-by-design of control systems that tolerate deterministic failures.
Weakly-hard constraints correspond to regular languages.
The size of the minimal finite state machine that recognizes whether a string satisfies the constraint (about $45k$ states for \anymiss{2}{300}) is a notorious impediment for the verification of control system properties.
This paper discusses an over-approximation of the language that allows us to provide sound safety guarantees for control systems under deadline misses that would be out of reach using the minimal finite state machine.
We present a compressed language acceptor and prove that it simulates the original finite state machine.
We study language cardinality properties, and report on empirical results that show how the new acceptor can be embedded in the control design workflow, leading to verifying safety for systems for which the state-of-the-art tools do not provide answers.  
\end{abstract}

\section{Introduction}
Modern Cyber-Physical Systems (CPSs) integrate dense sensing, actuation, and embedded computing under strict real-time requirements.
In domains such as autonomous driving, robotics, and industrial automation, increasing complexity makes timing coordination difficult, leading to sporadic deadline misses and delayed actuation.
As testified by a recent industrial survey it is common for real systems to have tasks missing deadlines: ``45\% of respondents indicated that the most time critical functions in the system could miss some deadlines, and 20\% indicated that deadline misses more often than 1 in 100 could be tolerated''~\cite{Akesson20}.
Even if infrequent, these timing problems can critically affect safety and stability of control systems, motivating the emergence of models that capture realistic failure patterns.
Weakly-hard constraints~\cite{Bernat01, Chao19} are widely used in industry~\cite{Hammadeh17b, Ahrendts18, Maggio20, Hammadeh20, Gallant25} to model deterministic timing failures, including deadline misses~\cite{Akesson20} and message losses~\cite{Linsenmayer17}.
The weakly-hard models bound not only how often faults occur, but also how they may be distributed over time, typically constructing admissible bit strings of successes ($1$) and failures ($0$).
The families of weakly-hard constraints are \anyhit{\weaklyhardh}{\weaklyhardk}, \anymiss{\weaklyhardm}{\weaklyhardk}, \rowhit{\weaklyhardh}{\weaklyhardk}, and \rowmiss{\weaklyhardm}{\weaklyhardk}.
For instance, \anymiss{2}{300} limits a control task to 2 deadline misses every 300 consecutive iterations, while \anyhit{8}{10} can enforce at least 8 correct message transfers in any 10 consecutive transmissions.
The \anymiss{\weaklyhardm}{\weaklyhardk} family is especially prominent in the verification of the stability of control systems under deadline misses and message losses~\cite{Pazzaglia18, Xu23, Vreman22b, Hobbs24, Seidel24, Hertneck21}. Weakly-hard windows that are found in industrial systems are often large (hundreds of samples)~\cite{Bernat01, Akesson20}, and the corresponding automata grow combinatorially with the size of the window~\cite{Vreman22a}.

The stability verification procedure uses two fundamental ingredients: the dynamics of the closed-loop system, and the weakly-hard pattern of how faults can appear, the latter represented as a finite automaton. The two ingredients are combined in a cross-product construction akin to what characterizes the automata-theoretic approach to model checking~\cite{DBLP:books/el/07/Vardi07,DBLP:reference/mc/Kupferman18}. The resulting structure is represented as a set of two matrices, of which the joint eigenvalue spectrum needs to be analyzed in order to prove system stability of the system in the presence of any of these fault patterns. The latter analysis however suffers from the size of the weakly-hard language acceptor that is one component of the cross-product. For instance, the language acceptor for $\anymiss{2}{300}$ has about $45k$ states. The problem can be alleviated by over-approximating the constraint by a more permissive constraint with a smaller window size. For instance, $\anymiss{2}{50}$ has just 1225 states and covers all fault patterns of $\anymiss{2}{300}$, but also many more. Stability under the smaller-window constraint therefore implies stability under the larger-window one, while instability in the smaller-window model does not translate back. In practice however, the additional fault patterns make it notoriously impossible to verify stability, since the additional failures introduce drift in the joint spectral characteristics. In fact, the joint eigenvalues all need to be below 1 to ensure system stability. This renders many real-world systems unverifiable in the presence of realistic weakly-hard constraints. 

The present paper contributes the following. Based on a very careful analysis of the structure of the minimum-size $\anymiss{\weaklyhardm}{\weaklyhardk}$ language acceptor, we come up with a family of automata $\trim{\weaklyhardm}{\weaklyhardk}{\comprate}$ with a configurable compression factor $\comprate$. We prove that each of them accepts an over-approximation of the language induced by $\anymiss{\weaklyhardm}{\weaklyhardk}$. The compressed automata can be drastically smaller without inducing a substantial drift of the joint eigenvalues. The latter is validated empirically on a set of 20 control benchmarks taken from the literature for which we report verification results, for all of which verification under large constraints ($k \approx 300$) is not feasible without over-approximation.
Note that the addressed situation is the following: large industrial window sizes $\weaklyhardk \approx 300$ and relatively small miss budgets $\weaklyhardm$, which is the practically relevant regime for weakly-hard control under rare deadline misses.

The paper is organized as follows.
Section~\ref{sec:background} introduces background knowledge and formalizes the problem statement.
Section~\ref{sec:compression} discusses the solution that we propose.
Section~\ref{sec:experiments} shows experimental evidence that our solution works well in practice, while Section~\ref{sec:conclusion} concludes the paper.
Appendices~\ref{sec:proofs},~\ref{sec:size} and~\ref{sec:cardinality} discuss additional matters and the proofs that are omitted in the text.

\section{Background and problem statement}
\label{sec:background}

CPS stability verification typically considers a model of a closed-loop system, which includes the physical dynamics and its  computational control model. We focus on linear dynamical systems, as is standard when linearizing around operating points~\cite{Astrom11, Maggio20}.
As done in the literature, the plant-controller evolution is modeled with two closed-loop matrices: $\clAh$ for nominal operation and $\clAm$ for faulty operation (e.g., induced by a deadline miss), extending the classic discrete-time linear time-invariant closed-loop system to 
\begin{equation}
\label{eqn:dyn}
x_{k+1}=\cl_{\sigma(k)}\, x_k, \quad\quad\quad \forall k, \sigma(k) \in \Sigma = \{ 0, 1 \}.
\end{equation}

A weakly-hard constraint determines how $\sigma(k)$ can vary over time, as a result of successions of task execution that hit or miss their deadlines. This is abstractly represented as task execution traces.

\begin{definition}[Task Execution Trace]
\label{def:inf-word}
The infinite word $(\traceitem_i)_{i \in \mathbb{N}^{+}}, \traceitem_i \in \{0, 1\}$ represents the execution trace of a task. 
\end{definition}

A value $\traceitem_i=0$ signifies that the $i$-th iteration of the task missed its deadline, while $\traceitem_i=1$ implies that the task hit the  $i$-th deadline. For $i>j$ we use $\traceitem_{i:j}$ to refer to the finite subword of $(\traceitem_i)_{i \in \mathbb{N}^{+}}$ spanned by the closed interval $[j,i]$. For example, when $\traceitem_{7:3} = 01101$ we have $\traceitem_7 = \traceitem_4 = 0$, while $\traceitem_6 =\traceitem_5 = \traceitem_3 = 1$. Note that indexing is in descending order.

The weakly-hard task model~\cite{Bernat01} (the \emph{de facto} standard to analyze systems that experience deadline misses) introduces four different types of tasks, two of which are relevant for this paper.

\begin{definition}[Weakly-Hard Task] \label{def:wh}
Given $\weaklyhardk \in \mathbb{N}^+$ and $\weaklyhardm, \weaklyhardh \in \{ n \in \mathbb{N}^+ \mid n < \weaklyhardk \}$, $\task$ is a weakly-hard task if it satisfies a constraint $\weaklyhardconstraint$, denoted $\task \vdash \weaklyhardconstraint$, which can be of one of the two following forms:
\begin{equation}
\begin{array}{c}
\task \vdash \anymiss{\weaklyhardm}{\weaklyhardk} \coloneq \forall i \in \mathbb{N}^{+}, \sum_{j=i}^{i+\weaklyhardk-1} (1-\traceitem_j) \leq \weaklyhardm, \\[2mm]
\task \vdash \anyhit{\weaklyhardh}{\weaklyhardk} \coloneq \forall i \in \mathbb{N}^{+}, \sum_{j=i}^{i+\weaklyhardk-1} \traceitem_j \geq \weaklyhardh.
\end{array}
\end{equation}
\end{definition}

For an $\anymiss{\weaklyhardm}{\weaklyhardk}$ constraint, in every window of $\weaklyhardk$ consecutive activations, the number of deadline misses is at most $\weaklyhardm$. For an  \anyhit{\weaklyhardm}{\weaklyhardk} constraint, in every window of $\weaklyhardk$ consecutive activations, the number of deadline hits is at least $\weaklyhardh$.

\begin{remark}
\label{lem:anyhitanymissequivalence}
The constraints are dual~\cite{Bernat01} in the sense that for $\weaklyhardh=\weaklyhardk-\weaklyhardm$,  
$$\task \vdash \anymiss{\weaklyhardm}{\weaklyhardk} \text{\qquad iff \qquad} \task \vdash \anyhit{\weaklyhardh}{\weaklyhardk}.$$
\end{remark}

\begin{remark}
\label{lem:minimum size}
Each weakly-hard constraint corresponds to a regular language and hence is recognized by a finite-state automaton $\originalautomaton$. The minimum automaton size needed is combinatorial in the window length: It comprises $\nicefrac{\weaklyhardk!}{\weaklyhardm!\,(\weaklyhardk-\weaklyhardm)!}$ states for \anymiss{\weaklyhardm}{\weaklyhardk} (and $\nicefrac{\weaklyhardk!}{\weaklyhardh!\,(\weaklyhardk-\weaklyhardh)!}$ for \anyhit{\weaklyhardh}{\weaklyhardk})~\cite{Vreman22}. The connection to the binomial coefficient $\binom{k}{m}$ is no coincidence. It is rooted in the need to distinguish all distinct ways of selecting $m$ zeroes in a window size of $k$.
\end{remark}

The link to Equation~(\ref{eqn:dyn}) is as follows. A given constraint \anymiss{\weaklyhardm}{\weaklyhardk} constrains how $\sigma(k)$ can vary over time, intuitively as a result of successions of task executions that hit or miss their deadlines. From the corresponding acceptor automaton $\originalautomaton$, adjacency matrices $\adjm$ and $\adjh$ for faulty executions (deadline misses) and nominal executions (deadline hits) can be constructed with the state count of both $\adjm$ and $\adjh$ being each equal to the automaton state count.

Combining the pair $(\clAm,\clAh)$ with the pair $(\adjm,\adjh)$, gives the lifted matrices $\adjm\otimes\,\clAm$ and $\adjh\otimes\,\clAh$ (where $\otimes$ is the Kronecker product), enabling stability verification via the joint spectral radius ($\text{jsr}$)~\cite{Rota60, Jungers09, Gripenberg96} (i.e.,~the eigenvalue with largest magnitude of the infinitely long matrix products of a set of matrices). The closed-loop system that experiences faulty executions is stable iff 
\begin{equation}
\label{eq:stability}
 \text{jsr}\,(\adjm\otimes\,\clAm, \adjh\otimes\,\clAh) < 1.
\end{equation}
This principle has substantially advanced the analysis of deadline misses in control tasks with demonstrated industrial impact~\cite{Vreman21, Gallant24, Gallant25}. It however suffers from the combinatorial growth of~$\originalautomaton$ and its contribution to the Kronecker product, if considering real world applicability.\footnote{Note that an alternative approach is to calculate the constrained joint spectral radius~\cite{Philippe16, Dercole17} and only build the products of matrices that correspond to viable systems evolutions, shifting the cause of the combinatorial growth problem from the matrices $(\adjm,\adjh)$ to the algorithmic choice of which evolutions to consider.}

\vspace{2mm}
\textbf{Notation.}
We overload the notation $\cardinality{\cdot}$ in three ways: if $S$ is a set, $\cardinality{S}$ denotes set cardinality. If $\originalautomaton$ is a finite automaton, the  $\cardinality{\originalautomaton}$ denotes its size (number of states). If $w$ is a finite word, then  $\cardinality{w}$ denotes its length.  We refer to the regular language accepted by a finite automaton $\originalautomaton$ as  $\lang\,(\originalautomaton)$. If $\lang$ is a language, we use the superscript $\lang^{=\ell}$ to restrict to  words of length $\ell$, thus $\lang^{=\ell} := \{ w \in \lang \mid \cardinality{w} = \ell \}$.

\subsection{Problem statement}
\label{sec:problem}

If we want to verify the stability of real-world examples, there are two important facts to consider:
\begin{enumerate*}[label=(\roman*)]
\item as discussed in the literature~\cite{Bernat01, Akesson20}, we would like to verify stability under constraints with sizes $\weaklyhardk \approx 300$, and 
\item usually the size of the system dynamics (represented as $\clAm$ and $\clAh$) is orders of magnitude smaller than that of $\adjm$ and $\adjh$.
\end{enumerate*}

In fact, the size of the automaton grows too fast with $\weaklyhardk$ (for \anymiss{3}{5} we build an automaton with $10$ states, while for \anymiss{3}{100} we reach $161$\,$700$ states, and for \anymiss{3}{300} the automaton has $4$\,$455$\,$100$ states).
This means that the limit of what can be verified are actually set by the size of the automaton that corresponds to the weakly-hard constraint.

Given the impossibility of determining whether the statement in Equation~\eqref{eq:stability} holds for large matrices due to scalability problems, we are thus striving for an over-approximation of the weakly-hard constraint to make the problem tractable.
Over-approximation here refers to the language accepted by the approximating automaton, relative to the acceptor automaton $\originalautomaton$. The approximand can contain more deadline misses than the original, but should be smaller in size.  

We now assume a certain weakly-hard constraint is given, with $\originalautomaton$ being its minimum-size acceptor automaton. The problem is to find an approximating automaton $\approximateautomaton$, such that:
\begin{enumerate}[label=(\roman*)]
\item $\lang\,(\originalautomaton) \subseteq \lang \, (\approximateautomaton)$,
\item $\cardinality{\approximateautomaton} \leq  \cardinality{\originalautomaton}$, and 
\item the stability verification problem~\eqref{eq:stability} is easier to solve with $\approximateautomaton$ than the original problem, also if considering other known approximations.
\end{enumerate}
The third item will mainly be studied  empirically, by examining  the verification effort needed in concrete scenarios, but we also provide analytical insights into this aspect. 

\section{Finding an automaton compression}
\label{sec:compression}

In this section we discuss our solution.
Our starting point is the minimum-size automaton that corresponds to an $\anymiss{\weaklyhardm}{\weaklyhardk}$, or equivalently $ \anyhit{\weaklyhardk-\weaklyhardm}{\weaklyhardk}$, weakly-hard constraint~\cite{Vreman22}.

We first establish an isomorphism on automaton nodes, for the purpose of transforming this automaton into one that provides us with information that we exploit for subsequent compression.
For the latter, we identify a set of nodes to remove from the automaton, together with rewiring its transitions.
Once we have performed these operations, we obtain a compressed automaton, denoted \trim{\weaklyhardm}{\weaklyhardk}{\comprate} where $\comprate \in \mathbb{N}^+$ is a compression factor.
We show that \trim{\weaklyhardm}{\weaklyhardk}{\comprate} simulates the minimal automaton that corresponds to the \anymiss{\weaklyhardm}{\weaklyhardk} weakly-hard constraint, regardless of the choice of $\comprate$. This  guarantees language inclusion, thus the first item of the problem defined in Section~\ref{sec:problem}. The second constraint is the consequence of the construction, which is parametric in $\comprate$. 

We now turn to $\aha = (\ahav, \ahae)$, the minimal automaton that corresponds to an \anymiss{\weaklyhardm}{\weaklyhardk} constraint, where $\ahav$ is the set of nodes and $\ahae$ the set of transitions. We shall provide a rigid definition soon, but to sharpen intuition we explain its construction~\cite{Vreman22} upfront.

Any node $\ahanode \in \ahav$ in the automaton is a word $\ahanode_{k:1}$ of length $k$, which corresponds to a window of the execution trace of length $k$ satisfying the weakly-hard constraint. Specifically we have the window positioned on the execution trace fragment $\traceitem_{k+i-1:i}$ after the completion of the $i$-th controller iteration.  When in $\ahanode_{k:1}$, reading symbol $0$ corresponds to the transition $(\ahanode_{k:1} \overset{\smash{0}}{\to} \ahanode_{k-1:1}\,0)$, while reading $1$ corresponds to the transition $(\ahanode_{k:1} \overset{\smash{1}}{\to} \ahanode_{k-1:1}\,1)$. Some of the transitions labelled $0$ however are impossible, because they violate the constraint. At the same time, some states represent more than one bitstring of length $k$, for the sake of minimality. For instance, when a transition moves the window to the right, it might be irrelevant what value moved out of the window, or it may not. In the former situations, the value is denoted $\digitx$.   Formally, we can define the automaton as follows, adapted from~\cite{Vreman22} (where it is presented as the result of automata minimisation).

\begin{definition}[Weakly-Hard Automaton for \anymiss{\weaklyhardm}{\weaklyhardk}]
The weakly-hard automaton $\aha = (\ahav, \ahae)$ can be written as
\begin{equation*}
\begin{array}{rcl}
\ahav &=& \{ \digitx^p \, \ahanode_{\weaklyhardk-p:1} | \ahanode_{\weaklyhardk-p:1} \in \{ 0, 1 \}^{\weaklyhardk-p} \land \ahanode_{\weaklyhardk-p}=\deadlinehit \land \\[2mm]
& & \phantom{\{\digitx^p \, \ahanode_{\weaklyhardk-p:1} |}\sum_i^{\weaklyhardk-p} v_i = \weaklyhardk-\weaklyhardm, \,\, \forall p \in \mathbb{N} \;|\; 0 \leq p \leq \weaklyhardm \},\\[2mm]
\ahae &=& \{(\ahanode_{\weaklyhardk:1}, i, \digitx^p\ahanode_{\weaklyhardk-p-1:1}i) \in \ahav \times \{0, 1\} \times \ahav \; | \; 0 \leq p \leq \weaklyhardm \}.
\end{array}
\end{equation*}
with initial node $\ahainitnode = \digitx^{\weaklyhardm}\,1^{\weaklyhardk-\weaklyhardm}$. All nodes are accepting.
\end{definition}
This automaton accepts all prefixes of the infinite-word language formed by the execution traces (Definition~\ref{def:inf-word}) satisfying the constraint.
Note that the underlying task-execution semantics are infinite-word semantics.
However, the weakly-hard constraint is prefix-closed: an infinite trace is admissible iff all its finite prefixes are admissible with respect to the sliding-window rule.
A corresponding Buchi automaton~\cite{Buchi1962}, accepting the infinite-word language is obtained by adding a trap state which is the only non-accepting, has loops for $1$ and $0$, and is reached by $0$ from any other state which is otherwise missing an outgoing zero-transition.
However, for the purpose of stability verification it is sufficient to describe all the evolutions that can manifest in the system behavior, and hence it is sufficient to use the finite prefixes automaton.

Figure~\ref{fig:example_automaton} shows the minimal automaton $\originalautomaton_{2,5}$ for the \anymiss{2}{5} constraint as an example.
The symbols $\digitx$ correspond to either zero or one.
If interpreting the $\digitx$ values as $0$, it is immediately obvious that the the states enumerate all possible options to place $2$ zeroes in a bitstring of length $5$, thus $5 \choose 2$ states result. 
In the following, it will be convenient to take $\digitx$ as corresponding to $0$ for calculations, and to keep in mind that the relevant portion of the node label is the suffix without any $\digitx$.

\begin{figure}[th]
\centering
\hspace{-5mm} 
\begin{tikzpicture}[>=stealth]
\graph[nodes={draw, rectangle, rounded corners, inner sep=4pt}, no placement,
    edge quotes={fill=white,inner sep=1pt}]
    {
        ini/\textcolor{white}{x}[white, at ={(-1.25,0)}];
    	XX111/$\phantom{\langle}\digitx\digitx111\phantom{\rangle}$[thick, double, double distance=0.05cm, at={(0,0)}];
		X1110/$\phantom{\langle}\digitx1110\phantom{\rangle}$[thick, double, double distance=0.05cm, at={(0,-1.5)}];
		X1101/$\phantom{\langle}\digitx1101\phantom{\rangle}$[thick, double, double distance=0.05cm, at={(6,-1.5)}];
		11100/$\phantom{\langle}11100\phantom{\rangle}$[thick, double, double distance=0.05cm, at={(0,-3)}, fill=black!10];
		11001/$\phantom{\langle}11001\phantom{\rangle}$[thick, double, double distance=0.05cm, at={(2,-3)}, fill=black!10];
    	10011/$\phantom{\langle}10011\phantom{\rangle}$[thick, double, double distance=0.05cm, at={(4,-3)}, fill=black!10];
		11010/$\phantom{\langle}11010\phantom{\rangle}$[thick, double, double distance=0.05cm, at={(6,-3)}, fill=black!10];
		10101/$\phantom{\langle}10101\phantom{\rangle}$[thick, double, double distance=0.05cm, at={(8,-3)}, fill=black!10];
		X1011/$\phantom{\langle}\digitx1011\phantom{\rangle}$[thick, double, double distance=0.05cm, at={(10,-1.5)}];
		10110/$\phantom{\langle}10110\phantom{\rangle}$[thick, double, double distance=0.05cm, at={(10,-3)}, fill=black!10];	

        {ini} -> XX111,
	    {XX111 [< "$1$"] } ->[loop above] XX111,
	    {XX111 [< "$0$"] } -> X1110,
	    {X1110 [< "$0$"] } -> 11100,
	    {X1110 [< "$1$"] } -> X1101,
	    {11100 [< "$1$"] } -> 11001,
	    {11001 [< "$1$"] } -> 10011,
	    {10011 [< "$1$"] } ->[bend right] XX111,
	    {X1101 [< "$0$"] } -> 11010,
	    {11010 [< "$1$"] } -> 10101,
	    {10101 [< "$1$"] } -> X1011,
	    {X1011 [< "$0$"] } -> 10110,
        {X1101 [< "$1$"] } -> X1011,
	    {10110 [< "$1$"] } -> X1101,
	    {X1011 [< "$1$"] } ->[bend right] XX111,
	    
    };
\end{tikzpicture}
\caption{Example: Automaton $\originalautomaton_{2,5}$ corresponding to \anymiss{2}{5}.}
\label{fig:example_automaton}
\end{figure}

\subsection{Automaton $\caha$ isomorphic to $\aha$}

We now present a transformation of the automaton (from $\aha$ to $\caha$) that will help introducing the compression algorithm.
We define a function $\indexfunc$ to determine the position of the $i$-th last appended zero in a string that grows on the right-hand side.

\begin{definition}[$\indexfunc$: index of the $i$-th last zero in $v$] \label{def indexfunc}
Function $\indexfunc : \ahav \times \mathbb{N}^{+} \rightarrow \mathbb{N}^{+} \cup \{\infty\}$ gives the index of the $i$-th last zero in $v$, or infinity when no such instance exists.
\begin{equation}
\begin{array}{rcl}
\indexfunc(\ahanode_{s:1}, i) \coloneq 
\begin{cases}
\infty & \sum_{j=1}^{s} (1-\ahanode_j) < i\\
\min_p \sum_{j=1}^{p} (1-\ahanode_j) = i & \text{otherwise}
\end{cases}
\end{array}
\end{equation}
\end{definition}
For example, for $\ahanode = \ahanode_{9:1} = 100111010$, we have $\indexfunc(\ahanode, 1) = 1$, $\indexfunc(\ahanode, 2) = 3$, $\indexfunc(\ahanode, 3) = 7$, $\indexfunc(\ahanode, 4) = 8$ and $\indexfunc(\ahanode, i) = \infty\,\,\forall i>4$.
Using $\indexfunc$, we define a function $\mapfunc$ to rename the nodes and highlight some of their properties.

\begin{definition}[Node mapping] \label{def mapfunc}
Function $\mapfunc : \ahav \times \mathbb{N}^{+} \times \mathbb{N}^{+} \rightarrow \cahav$ maps the nodes in $\ahav$ and the characteristics $\weaklyhardm$ and $\weaklyhardk$ of an \anymiss{\weaklyhardm}{\weaklyhardk} weakly hard constraint into a new set of nodes 
\begin{equation}
\cahav = \left\{ \langle \cahanode_1, \dots, \cahanode_{\weaklyhardm} \rangle \mid 0 \leq \cahanode_i \leq \weaklyhardk - \weaklyhardm \,\land\, \cahanode_i \geq \cahanode_j \text{ for } i < j \right\}
\end{equation}
according to
\begin{equation}
\mapfunc(\ahanode_{k:1}, \weaklyhardm, \weaklyhardk) \coloneq
\langle \cahanode_1, \dots, \cahanode_{\weaklyhardm} \rangle \mid \cahanode_i = \weaklyhardk - \weaklyhardm + i - \indexfunc(\ahanode_{k:1}, i).
\end{equation}
\end{definition}
Having transformed the values of $\digitx$ into zeros ensures that in each label there are exactly $\weaklyhardm$ zeros.
Intuitively, $\cahanode_1$ counts the number of ones that needs to be appended to the original node label to reach a node in the automaton with no zeros in the relevant part of the node label (i.e., the substring that does not contain any $\digitx$) and $\cahanode_2$ counts the number of ones that needs to be appended to reach a state with at most a single zero in the relevant part of the string.
More in general, $\cahanode_i$ counts the number of ones that needs to be appended to reach a node in which at most $i-1$ zeros are present in the relevant part of the string.

We now consider a fixed value for $\weaklyhardm$ and $\weaklyhardk$ and show that, given these values, $\mapfunc(\cdot,\weaklyhardm,\weaklyhardk)$ is a bijective function, i.e., that there is a one-to-one correspondence between elements of $\ahav$ and elements of $\cahav$.

\begin{theorem} \label{f bijective}
Given $\weaklyhardm$ and $\weaklyhardk$, the function $\mapfunc(\cdot,\weaklyhardm,\weaklyhardk): \ahav \rightarrow \cahav$ is bijective.
\begin{proof}
  See Appendix~\ref{proof f bijective}.  
\end{proof}
\end{theorem}

After defining $\cahav$, we now define a set $\cahae$ of transitions between these nodes.
\begin{definition}[Set of transitions] \label{def transitions}
The set of transitions $\cahae$ allows to read the symbol $1$ from all the nodes in the graph and the symbol $0$ only when introducing said $0$ would not violate the \anymiss{\weaklyhardm}{\weaklyhardk} constraint. Formally,
\begin{align*}
\cahae = & \left\{( \langle \cahanode_1, \dots, \cahanode_{\weaklyhardm} \rangle, \deadlinehit, \langle \max(\cahanode_1 - 1, 0), \dots, \max(\cahanode_{\weaklyhardm} - 1, 0) \rangle) \mid \cahanode \in \cahav \right\} \\
& \cup \left\{( \langle \cahanode_1, \dots, \cahanode_{\weaklyhardm} \rangle, \deadlinemiss, \langle \weaklyhardk-\weaklyhardm, \cahanode_1, \dots, \cahanode_{\weaklyhardm - 1} \rangle) \, \mid \, \cahanode \in \cahav \,\land\, \cahanode_{\weaklyhardm} = 0 \right\}.
\end{align*}
\end{definition}

\begin{lemma} \label{F lambda well defined}
$\cahae$ is well-defined, i.e., $\cahae \subseteq \cahav \times \alphabet \times \cahav$.
\begin{proof}
    See Appendix~\ref{proof F lambda well defined}.
\end{proof}
\end{lemma}

Let $\caha = (\cahav, \cahae)$.
All nodes are accepting and the initial node $\cahainitnode$ is the tuple with arity $\weaklyhardm$ that includes only zeros, $\langle 0, \dots, 0 \rangle$.
We conclude this subsection by showing that $\aha$ and $\caha$ equivalently represent the weakly-hard constraint, i.e., they are isomorphic.

\begin{theorem} \label{G and H isomorph}
$\aha$ and $\caha$ are isomorphic.

\begin{proof}
See Appendix~\ref{isomorphism proof}.
\end{proof}
\end{theorem}

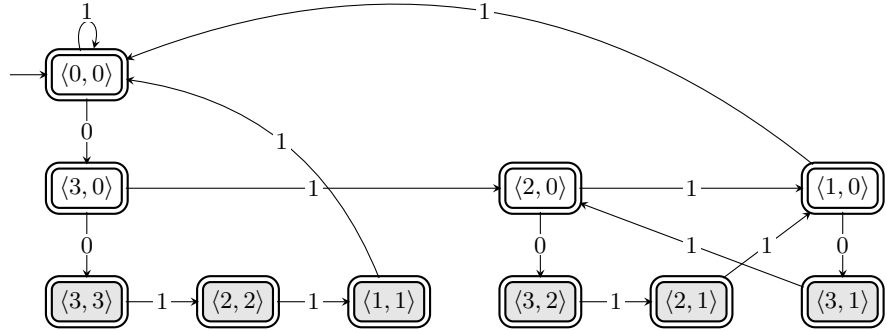
\begin{figure}[ht]
\centering
\hspace{-5mm} 
\begin{tikzpicture}[>=stealth]
\graph[nodes={draw, rectangle, rounded corners, inner sep=4pt}, no placement,
    edge quotes={fill=white,inner sep=1pt}]
    {
        ini/\textcolor{white}{x}[white, at ={(-1.25,0)}];
    	XX111/{$\langle 0,0 \rangle$}[thick, double, double distance=0.05cm, at={(0,0)}];
		X1110/{$\langle 3,0 \rangle$}[thick, double, double distance=0.05cm, at={(0,-1.5)}];
		X1101/{$\langle 2,0 \rangle$}[thick, double, double distance=0.05cm, at={(6,-1.5)}];
		11100/{$\langle 3,3 \rangle$}[thick, double, double distance=0.05cm, at={(0,-3)}, fill=black!10];
		11001/{$\langle 2,2 \rangle$}[thick, double, double distance=0.05cm, at={(2,-3)}, fill=black!10];
    	10011/{$\langle 1,1 \rangle$}[thick, double, double distance=0.05cm, at={(4,-3)}, fill=black!10];
		11010/{$\langle 3,2 \rangle$}[thick, double, double distance=0.05cm, at={(6,-3)}, fill=black!10];
		10101/{$\langle 2,1 \rangle$}[thick, double, double distance=0.05cm, at={(8,-3)}, fill=black!10];
		X1011/{$\langle 1,0 \rangle$}[thick, double, double distance=0.05cm, at={(10,-1.5)}];
		10110/{$\langle 3,1 \rangle$}[thick, double, double distance=0.05cm, at={(10,-3)}, fill=black!10];	

        {ini} -> XX111,
	    {XX111 [< "$1$"] } ->[loop above] XX111,
	    {XX111 [< "$0$"] } -> X1110,
	    {X1110 [< "$0$"] } -> 11100,
	    {X1110 [< "$1$"] } -> X1101,
	    {11100 [< "$1$"] } -> 11001,
	    {11001 [< "$1$"] } -> 10011,
	    {10011 [< "$1$"] } ->[bend right] XX111,
	    {X1101 [< "$0$"] } -> 11010,
	    {11010 [< "$1$"] } -> 10101,
	    {10101 [< "$1$"] } -> X1011,
	    {X1011 [< "$0$"] } -> 10110,
        {X1101 [< "$1$"] } -> X1011,
	    {10110 [< "$1$"] } -> X1101,
	    {X1011 [< "$1$"] } ->[bend right] XX111,
	    
    };
\end{tikzpicture}
\caption{Example: Automaton $\isomorphautomaton_{2,5}$ for $\anymiss{2}{5}$.}
\label{fig:example_automaton_renamed}
\end{figure}

We have now shown that we can construct an equivalent automaton $\caha$, that includes information that will be used for the compression introduced in the following.
An example of the construction is given in Figure~\ref{fig:example_automaton_renamed}, where we transform the automaton presented in Figure~\ref{fig:example_automaton}.
To simplify the wording in the following, we introduce two classes of nodes:
\begin{itemize}
\item the critical nodes, which do not admit a zero-transition, i.e., $\cahanode_{\weaklyhardm} \neq 0$, and
\item the non-critical nodes, from which a zero-transition is allowed, i.e., $\cahanode_{\weaklyhardm} = 0$.
\end{itemize}
In Figure~\ref{fig:example_automaton_renamed}, the critical nodes are the nodes with gray background and the non-critical nodes are the nodes with white background.

\subsection{Automaton $\compaha$ as the compression of $\caha$}

We now look for a way to compress $\caha$, in which we only merge critical nodes, which we consider ``similar'' because they all have the property that they do not allow a zero-transition.
Specifically, in building \trim{\weaklyhardm}{\weaklyhardk}{\comprate} we merge a maximum of $\comprate$ nodes according to their requirements in terms of one-transitions to reach specific non-critical nodes.
For each number $i \in \mathbb{N} \mid 1 \leq i \leq \weaklyhardm - 1$, we merge up to $\comprate \in \mathbb{N}^+$ nodes that require the same numbers of one-transitions to reach a node that has at most $i$ zeros in the relevant part.
For example, in Figure~\ref{fig:example_automaton_renamed} critical nodes are considered similar if they require the same number of one-transitions to reach a node with at most one zero, since $\weaklyhardm - 1 = 1$.
This brings us to state that nodes $\langle 1, 1 \rangle$, $\langle 2, 1 \rangle$ and $\langle 3, 1 \rangle$ are similar because with a single one-transition they reach a non-critical node.
Similarly, we consider nodes $\langle 2, 2 \rangle$ and $\langle 3, 2 \rangle$ similar because two one-transitions from these nodes will take the automaton to a non-critical node.

When merging nodes to build \trim{\weaklyhardm}{\weaklyhardk}{\comprate}, one needs to decide which outgoing transition are to be kept.
In the example of Figure~\ref{fig:example_automaton_renamed}, $\langle 2, 2 \rangle$ has a one-transition to $\langle 1, 1 \rangle$ and $\langle 3, 2 \rangle$ has a one-transition to $\langle 2, 1 \rangle$.
Since we want the \trim{\weaklyhardm}{\weaklyhardk}{\comprate} automaton to be an over-approximation of the original \anymiss{\weaklyhardm}{\weaklyhardk} automaton, every word accepted by the \anymiss{\weaklyhardm}{\weaklyhardk} automaton should also be accepted by \trim{\weaklyhardm}{\weaklyhardk}{\comprate}.
Therefore, in a merged node, we need to keep the transitions that allow for a higher number of zeros.
Looking at the example, node $\langle 2, 2 \rangle$ allows two zero-transitions after two one-transitions, whereas $\langle 3, 2 \rangle$ allows only one zero-transition after two one-transitions.
Hence, we keep the outgoing one-transition from node $\langle 2, 2 \rangle$ to $\langle 1, 1 \rangle$.
In the case of the automaton for \anymiss{\weaklyhardm}{\weaklyhardk}, allowing a higher number of zeros corresponds to the node label with smaller numbers.
This implies that we keep the transition to the node with the smallest numbers.

\begin{figure}[t]
\centering
\hspace{-5mm} 
\begin{tikzpicture}[>=stealth]
\graph[nodes={draw, rectangle, rounded corners, inner sep=4pt}, no placement,
    edge quotes={fill=white,inner sep=1pt}]
    {
        ini/\textcolor{white}{x}[white, at ={(-1.25,0)}];
    	XX111/{$\langle 0,0 \rangle$}[thick, double, double distance=0.05cm, at={(0,0)}];
		X1110/{$\langle 3,0 \rangle$}[thick, double, double distance=0.05cm, at={(0,-1.5)}];
		X1101/{$\langle 2,0 \rangle$}[thick, double, double distance=0.05cm, at={(2,-1.5)}];
		11100/{$\langle 3,3 \rangle$}[thick, double, double distance=0.05cm, at={(0,-3)}, fill=black!10];
		11001/{$\langle 2,2 \rangle$}[thick, double, double distance=0.05cm, at={(2,-3)}, fill=black!10];
    	10011/{$\langle 1,1 \rangle$}[thick, double, double distance=0.05cm, at={(4,-3)}, fill=black!10];
		X1011/{$\langle 1,0 \rangle$}[thick, double, double distance=0.05cm, at={(4,-1.5)}];

        {ini} -> XX111,
	    {XX111 [< "$1$"] } ->[loop above] XX111,
	    {XX111 [< "$0$"] } -> X1110,
	    {X1110 [< "$0$"] } -> 11100,
	    {X1110 [< "$1$"] } -> X1101,
	    {11100 [< "$1$"] } -> 11001,
	    {11001 [< "$1$"] } -> 10011,
	    {10011 [< "$1$"] } ->[bend right] XX111,
	    {X1101 [< "$0$"] } -> 11001,
	    {X1011 [< "$0$"] } -> 10011,
        {X1101 [< "$1$"] } -> X1011,
	    {X1011 [< "$1$"] } ->[bend right] XX111,
	    
    };
\end{tikzpicture}
\caption{Example: $\approximateautomaton_{2,5,3}$ for \trim{2}{5}{3}.}
\label{fig:example_automaton_compressed}
\end{figure}

Figure~\ref{fig:example_automaton_compressed} shows the \trim{\weaklyhardm}{\weaklyhardk}{\comprate} automaton of Figure~\ref{fig:example_automaton_renamed} after performing the described compression.
Compressing the automaton in this way reduces the number of nodes but at the same time increases the number of accepted words. For example, the compressed automaton accepts the word $\deadlinemiss \deadlinehit \deadlinehit \deadlinemiss \deadlinehit \deadlinemiss \deadlinemiss$, which is not accepted by the non-compressed automaton.
In the figure, all similar nodes have been merged (meaning that $\comprate = 3$).
However, it is also possible to compress less by defining a maximum number of nodes that are merged.
For example, if this number is limited to two, only the nodes $\langle 2, 2 \rangle$ and $\langle 3, 2 \rangle$ as well as $\langle 1, 1 \rangle$ and $\langle 2, 1 \rangle$ in Figure~\ref{fig:example_automaton_renamed} would be merged and $\langle 3, 1 \rangle$ would remain.
Setting $\comprate=1$ makes the non-compressed automaton for \anymiss{\weaklyhardm}{\weaklyhardk} a special case of \trim{\weaklyhardm}{\weaklyhardk}{\comprate}.

We now restrict ourself to the case in which $\weaklyhardm > 1$, otherwise \trim{\weaklyhardm}{\weaklyhardk}{\comprate} and \anymiss{\weaklyhardm}{\weaklyhardk} will produce the same automaton. We define the compressed automaton $\compaha = (\compahav, \compahae)$ for \trim{\weaklyhardm}{\weaklyhardk}{\comprate}.
\begin{definition}[Compressed automaton
]
We define the compressed automaton $\compaha = (\compahav, \compahae)$ as
\begin{align*}
\compahav = & \,\,\,\{ \langle \cahanode_1, \dots, \cahanode_{\weaklyhardm} \rangle \mid \, 0 \leq \cahanode_i \leq \weaklyhardk - \weaklyhardm \,\land\, \cahanode_i \geq \cahanode_j \text{ for } i < j \\ 
& \phantom{\,\,\,\{ \langle \cahanode_1, \dots, \cahanode_{\weaklyhardm} \rangle \mid \,}
\land\, ((\cahanode_1 - \cahanode_2) \text{ mod } \comprate = 0 \text{ or } \cahanode_{\weaklyhardm} = 0) \},\\
\compahae = & \,\,\,\{( \langle \cahanode_1 \dots, \cahanode_{\weaklyhardm} \rangle, \deadlinehit, \langle \max(\cahanode_1 - 1, 0), \dots, \max(\cahanode_{\weaklyhardm} - 1, 0) \rangle) \mid \cahanode \in \compahav \} \\
& \,\,\,\cup \{( \langle \cahanode_1, \dots, \cahanode_{\weaklyhardm} \rangle, \deadlinemiss, \langle \cahanode_0, \cahanode_1, \dots, \cahanode_{\weaklyhardm - 1} \rangle) \mid  \cahanode \in \compahav \,\land\, \cahanode_{\weaklyhardm} = 0  \\
& \quad \,\, \,\land\, \cahanode_0 =
\begin{cases}
\weaklyhardk - \weaklyhardm \qquad \qquad \qquad \qquad \qquad \quad\,\, \text{ if } \cahanode_{\weaklyhardm - 1} = 0 \\
\weaklyhardk - \weaklyhardm - (\weaklyhardk - \weaklyhardm - \cahanode_1) \text{ mod } \comprate \quad \quad \text{otherwise}
\end{cases} \}.
\end{align*}
All nodes are accepting, and the initial node is $\langle 0, \dots, 0 \rangle$.
\end{definition}

\begin{lemma} \label{F lambda m well defined}
$\compahae$ is well-defined, i.e., $\compahae \subseteq \compahav \times \alphabet \times \compahav$.
\begin{proof}
See Appendix~\ref{proof F lambda m well defined}.
\end{proof}
\end{lemma}

Now that we have defined a candidate compression, we want to prove that the language recognized by the automaton that corresponds to a given \anymiss{\weaklyhardm}{\weaklyhardk} is a subset of the language recognized by the \trim{\weaklyhardm}{\weaklyhardk}{\comprate} automaton, i.e., $\lang\,(\aha) \subseteq \lang\,(\compaha)$ (the first point in our problem statement from Section~\ref{sec:problem}).
To do so, we show that the compressed automaton $\compaha = (\compahav, \compahae)$ simulates $\caha = (\cahav, \cahae)$, which is isomorphic to $\aha$.

\begin{theorem} \label{simulation argument}
$\compaha$ simulates $\caha$ via
\begin{align*}
\binaryrel = \{ ( \langle \cahanode_1, \dots, \cahanode_{\weaklyhardm} \rangle, \langle \hat{\cahanode}_1, \dots, \hat{\cahanode}_{\weaklyhardm} \rangle ) \mid \cahanode \in \cahav, \hat{\cahanode} \in \compahav &\,\land\, \cahanode_i \geq \hat{\cahanode}_i \}.
\end{align*}
We write $\cahanode \binaryrelsignlong \hat{\cahanode}$ or simply $\cahanode \binaryrelsign \hat{\cahanode}$ to express that $(\cahanode,\hat{\cahanode})\in \binaryrel$.

\begin{proof}
$\compaha$ simulates $\caha$ via $\binaryrel$ if the following conditions hold~\cite[p. 13]{Park81}:
\begin{enumerate} [label = \emph{(\roman*)}]
\item $\binaryrel \subseteq \cahav \times \compahav$
\item $\cahainitnode \binaryrelsign \hat{\cahanode}_{init}$
\item If $\cahanode$ is an accepting node and $\cahanode \binaryrelsign \hat{\cahanode}$, then $\hat{\cahanode}$ is also an accepting node.
\item If $\cahanode \binaryrelsign \hat{\cahanode}$, $\alphabetletter \in \alphabet$ and $(\cahanode, \alphabetletter, \cahanode') \in \cahae$, then there exists a $\hat{\cahanode}'$ such that $(\hat{\cahanode}, \alphabetletter, \hat{\cahanode}') \in \compahae$ and $\cahanode' \binaryrelsign \hat{\cahanode}'$.
\end{enumerate}
\begin{enumerate} [label = \emph{(\roman*)}]
\item{
This follows directly from the definition of $\binaryrel$.
}
\item{
Obviously, $\cahainitnode = \langle 0, \dots, 0 \rangle \binaryrelsign \langle 0, \dots, 0 \rangle = \hat{\cahanode}_{init}$ is valid.
}
\item{
Since all nodes of $\compaha$ and $\caha$ are accepting, this condition is satisfied.
}
\item{
Let $\cahanode \binaryrelsign \hat{\cahanode}$ with $\cahanode = \langle \cahanode_1, \dots, \cahanode_{\weaklyhardm} \rangle$ and $\hat{\cahanode} = \langle \hat{\cahanode}_1, \dots, \hat{\cahanode}_{\weaklyhardm} \rangle$. Then $\cahanode \in \cahav$, $\hat{\cahanode} \in \compahav$, and $\cahanode_i \geq \hat{\cahanode}_i$.

First, let $\alphabetletter = \deadlinemiss$ and $(\cahanode, \deadlinemiss, \cahanode') \in \cahae$. Then $\cahanode_{\weaklyhardm} = 0$. Since $\cahanode_{\weaklyhardm} \geq \hat{\cahanode}_{\weaklyhardm}$, also $\hat{\cahanode}_{\weaklyhardm} = 0$ and there exists a transition $(\hat{\cahanode}, \deadlinemiss, \hat{\cahanode}') \in \compahae$. We have
\begin{align*}
\cahanode' &= \langle \weaklyhardk-\weaklyhardm, \cahanode_1, \dots, \cahanode_{\weaklyhardm - 1} \rangle \\
\hat{\cahanode}' &= 
\begin{cases}
\langle \weaklyhardk-\weaklyhardm, \hat{\cahanode}_1, \dots, \hat{\cahanode}_{\weaklyhardm - 1} \rangle &\text{if } \hat{\cahanode}_{\weaklyhardm - 1} = 0 \\
\langle \weaklyhardk-\weaklyhardm - (\weaklyhardk-\weaklyhardm - \hat{\cahanode}_1) \text{ mod } \comprate, \hat{\cahanode}_1, \dots, \hat{\cahanode}_{\weaklyhardm - 1} \rangle &\text{otherwise.}
\end{cases}
\end{align*}
$\cahanode' \in \cahav$ and $\hat{\cahanode}' \in \compahav$ follow from the well-definedness of $\cahae$ and $\compahae$, see Lemma \ref{F lambda well defined} and \ref{F lambda m well defined}. For $i > 1$, $\cahanode'_i = \cahanode_{i - 1} \geq \hat{\cahanode}_{i - 1} = \hat{\cahanode}'_i$ . For $i = 1$, $\cahanode'_1 \geq \hat{\cahanode}'_1$ is valid in the first case because $\weaklyhardk-\weaklyhardm \geq \weaklyhardk-\weaklyhardm$, and in the second case since $\weaklyhardk-\weaklyhardm \geq \weaklyhardk-\weaklyhardm - (\weaklyhardk-\weaklyhardm - \hat{\cahanode}_1) \text{ mod } \comprate$, see Inequality \ref{modulo leq applied}.

Next, let $\alphabetletter = \deadlinehit$ and $(\cahanode, \deadlinehit, \cahanode') \in \cahae$. Since there exists a one-transition in $\compahae$ for every node in $\compahav$, there exists a transition $(\hat{\cahanode}, \deadlinehit, \hat{\cahanode}') \in \compahae$. We have
\begin{align*}
\cahanode' &= \langle \max(\cahanode_1 - 1, 0), \dots, \max(\cahanode_{\weaklyhardm} - 1, 0) \rangle \\
\hat{\cahanode}' &= \langle \max(\hat{\cahanode}_1 - 1, 0), \dots, \max(\hat{\cahanode}_{\weaklyhardm} - 1, 0) \rangle.
\end{align*}
Again, $\cahanode' \in \cahav$ and $\hat{\cahanode}' \in \compahav$ follow from the well-definedness of $\cahae$ and $\compahae$. $\cahanode'_i = \max(\cahanode_i - 1, 0) \geq \max(\hat{\cahanode}_i - 1, 0) = \hat{\cahanode}'_i$ because of $\cahanode_i \geq \hat{\cahanode}_i$.
}
\end{enumerate}
\end{proof}
\end{theorem}
Simulation implies trace inclusion~\cite{DBLP:conf/birthday/LarsenSS22} which agrees with language inclusion because all states are accepting. 
With the isomorphism result (Theorem~\ref{G and H isomorph}) we have established the first item in our problem statement.

\begin{corollary}
$\lang\,(\aha) \subseteq \lang \, (\compaha)$.
\end{corollary}

\begin{remark}\label{rem:26}
We are later going to use the following insight. For large values of $\wordlength$, the number of accepted words of length $\wordlength$, i.e., $\cardinality{ \langlength{\lang} }$, can be approximated by an exponential function of the form $a \cdot b^{\,\wordlength}$. The values of $a$ and $b$ can be derived  from the shape of the respective automaton, and this holds for all acceptors $\mathcal G$ we consider.  

Concretely, assuming  $\mathcal G=(V,E)$ is of size $\cardinality {V} = n$, we let  $\initstatevector \in \realnumbersmatrix{1}{n}$ denote a vector representing the initial state $\ahainitnode$ (i.e.,~$\initstatevector_i = 1$ if $\ahanode_i$ is the initial state and $0$ otherwise), and $\eigenvalue$ be the largest eigenvalue. Denote with $\eigenvectory$ the eigenvector that corresponds to the eigenvalue $\eigenvalue$ in the adjacency matrix defined by the entries $\adjelement_{ij} = 1$ if $\{ (\ahanode_i, \deadlinehit, \ahanode_j) , (\ahanode_i, \deadlinemiss, \ahanode_j)\} \cap  E \neq \emptyset$, and $0$ otherwise. Let $\eigenvectorx$ be the eigenvector corresponding to the largest eigenvalue of the transposed adjacency matrix. Then, 

\begin{align} \label{card ev approx}
\cardinality{ \langlength{\lang}\,({\mathcal G}) } \approx \onenorm{ \frac{\eigenvectory^\top \initstatevector^\top}{\eigenvectory^\top \eigenvectorx} \eigenvectorx} \cdot \eigenvalue^\wordlength
\end{align}
where the $1$-norm is defined as usual. The proof is found in Appendix~\ref{sec:cardinality}.
\end{remark}

\subsection{Number of states of $\compaha$}

We move on to the second item of the problem statement, i.e., the number of states of the compressed automaton $\compaha = (\compahav, \compahae)$.
To calculate exactly the number of states of $\compaha$, i.e., $\cardinality{\compahav}$, we split the set of states into subsets.
We start by defining the subset of critical nodes, i.e., nodes from which a zero-transition will violate the constraint and thus is not allowed, $\critahav = \{\cahanode \in \compahav \mid \cahanode_{\weaklyhardm} \neq 0 \}$.
The complement set $\noncritahav = \compahav \setminus \critahav = \{\cahanode \in \compahav \mid \cahanode_{\weaklyhardm} = 0 \}$, comprises all nodes from which a zero-transition is allowed.
Clearly,
\begin{align} \label{card m}
\cardinality{\compahav} = \cardinality{\critahav} + \cardinality{\noncritahav}.
\end{align}

The cardinality of $\noncritahav$ is rather easy to determine, since by definition
\begin{align*}
\noncritahav = \{ \langle \cahanode_1, \dots, \cahanode_{\weaklyhardm} \rangle \mid 0 \leq \cahanode_i \leq \weaklyhardk-\weaklyhardm \text{ and } \cahanode_i \geq \cahanode_j \text{ for } i < j \text{ and } \cahanode_{\weaklyhardm} = 0 \}.
\end{align*}

The last element of the $\weaklyhardm$-arity tuples that constitute the node label is fixed to zero.
This means that the set $\noncritahav$ is built by selecting $\weaklyhardm - 1$ elements from a set with  $\weaklyhardk-\weaklyhardm + 1$ total elements, giving us
\begin{align} \label{card m noncrit final}
\cardinality{\noncritahav} 
= \left( \!\! {\weaklyhardk-\weaklyhardm + 1 \choose \weaklyhardm - 1} \!\! \right)
= {{\weaklyhardk - 1} \choose {\weaklyhardm - 1}}.
\end{align}
The cardinality of $\critahav$ is more complicated to determine and requires to enumerate three possible cases
, eventually giving us the following result.

\begin{theorem}\label{proof size}
The cardinality of the compressed automaton \trim{\weaklyhardm}{\weaklyhardk}{\comprate} is
\begin{equation}
\cardinality{\compahav} = {\weaklyhardk - 1 \choose \weaklyhardm - 1} + 
\begin{cases}
   %
   \weaklyhardk-\weaklyhardm \vphantom{\sum\limits_{j = 1}^{\weaklyhardk-\weaklyhardm}}
      & \quad\weaklyhardm = 1 \\[3mm]
   \sum\limits_{j = 1}^{\weaklyhardk-\weaklyhardm} \left\lceil \frac{j}{\comprate} \right\rceil
      & \quad\weaklyhardm = 2 \\[3mm]
   \sum\limits_{i = 0}^{\weaklyhardk-\weaklyhardm - 1} 
     {\weaklyhardm - 3 + i \choose i} \sum\limits_{j = 1}^{\weaklyhardk-\weaklyhardm - i} \left\lceil \frac{j}{\comprate} \right\rceil 
      & \quad\weaklyhardm > 2
\end{cases}.
\end{equation}

\begin{proof}
See Appendix~\ref{sec:size}.
\end{proof}
\end{theorem}

As an example, Figure~\ref{fig:cardinality} displays the number of states of \trim{2}{300}{\comprate} for varying $c$, all of which over-approximate the \anymiss{2}{300} constraint. The latter is considered a realistic requirement in practical contexts. 

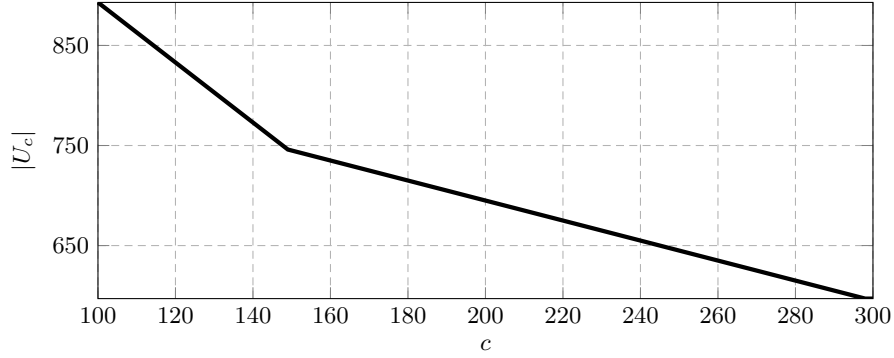
\begin{figure}[ht]
\centering
\begin{tikzpicture}
\begin{axis}[width=0.97\textwidth, height=5.5cm,
xlabel={$\comprate$}, ylabel={$\cardinality{\compahav}$}, enlargelimits=false,
ytick={650,750,850}, grid=major, grid style={black!30,densely dashed}]
\addplot[ultra thick, no marks] table[col sep=comma, x=count, y=size_compressed] {data/cardinality.csv};
\end{axis}
\end{tikzpicture}
\caption{Illustration of the size of \trim{2}{300}{\comprate} when $\comprate \in [100,300]$.}
\label{fig:cardinality}
\end{figure}

\begin{remark} 
Since we are interested in relatively large values of $\weaklyhardk$ and small values of $\weaklyhardm$, the above result is powerful, even though the size is lower bounded by $\weaklyhardk-1 \choose \weaklyhardm-1$ which appears not far off from $\weaklyhardk \choose \weaklyhardm$. Taking $\anymiss{2}{300}$ as an example, we see that ${300 \choose 2} = 44\,850$ while  ${299 \choose 1} = 299$. This compares favorably to the other known over-approximation $\anymiss{\weaklyhardm}{\weaklyhardk-1}$, or concretely $\anymiss{2}{299}$  which has ${299 \choose 2} = 44\,551$ states. To reach the area of the plot in Figure~\ref{fig:cardinality} we would need to work with $\anymiss{2}{40}$ having 780 states. This admits a lot more deadline misses than the original constraint, thus making it difficult to verify stability. We return to this aspect in Section~\ref{sec:experiments}.
\end{remark}

\section{Empirical results}
\label{sec:experiments}

The core task the verification engine needs perform is to assert Equation~\eqref{eqn:dyn}.
In line with the third item in our problem statement, we are now looking at the practical relevance of our findings, so as to address the third aspect of our problem statement. For this, we study the following research questions.
\begin{itemize}[leftmargin=1cm]
\item [\textbf{RQ1:}] Can we solve stability problems that are beyond reach of current methods?
\item [\textbf{RQ2:}] Is the over-approximation more efficient than weakly-hard constraints at solving stability problems for a large class of systems and by how much?
\item [\textbf{RQ3:}] How does the compression rate $\comprate$ influence the verification efficiency?
\end{itemize}
Each question will be addressed by a dedicated experimental campaign.
Note that the current methods mentioned in \textbf{RQ1} imply the introduction of (different) over-approximations for tractability. On the contrary, \textbf{RQ2} focuses on the comparison between verification under exact weakly-hard constraints and our over-approximation and stresses the importance of a larger test campaign, that allows us to see if the benefits are visible on a large family of physical systems.
\textbf{RQ3} then investigate the effect of the only parameter of our over-approximation method.

\vspace{2mm}
As testbench, we use a set of $20$ plants taken from the literature. 
Specifically, we use the test plant $P_1(s) = e^{-s} / (1+s\,T)$ from a collection of control benchmarks typically used to evaluate the tuning of controllers, for further details on the test batch, see \cite[Chapter 7, Equation (7.2)]{Astrom06}.
We test the plant with the proposed different values of $T \in \{ 0.05, 0.1, 0.2, 0.3, 0.5,$ $0.7, 1, 1.3, 1.5, 2, 4, 6, 10, 20, 50, 100, 200, 500, 1000 \}$.
For each of the $20$ plants, we design a Luenberger observer with closed-loop poles chosen randomly in the interval $[0.2, 0.3]$ and a state-feedback controller that imposes closed-loop poles chosen randomly in the interval $[0.6, 0.7]$ using standard control methods~\cite{Astrom11}.
The randomness simplifies the solution of the control problem by avoiding the introduction of multiple singularities at the same location.

The combination of plant, observer, and controller allows us to calculate the matrices $\clAm$ and $\clAh$ in Equation~\eqref{eq:stability}, where $\clAh$ represents the system dynamics when a deadline is hit by the controller, and $\clAm$ encodes the behavior in case of deadline miss.\footnote{The size of $\clAh$ and of $\clAm$ varies across benchmarks, but is at most $10\!\times\!10$.} Note that the controller is specifically implemented with one-step delay, or according to the Logical Execution Time paradigm~\cite{Kirsch2011}, meaning that the control signal calculated in the $i$-th iteration is applied during the $(i+1)$-th iteration, making the spectral radius of $\clAh$ move from the original design choice. We assume that the control task is killed when a deadline is missed and that the control signal is kept constant, but the approach works for any possible miss-handling strategy.

To find bounds on the joint spectral radius, we use the branch-and-bound algorithm proposed by Gripenberg~\cite{Gripenberg96}. The algorithm is iteratively narrowing down an interval in which the joint spectral radius would lie. For obvious reasons, its lower bound corresponds to the one of the system without deadline misses, i.e., to the spectral radius of $\clAh$, $\text{sr}\,(\clAh)$. At each iteration, a new set of matrices is computed, effectively pruning paths that are identified as irrelevant. In this way, the upper bound on the joint spectral radius is lowered. As soon as the upper bound is below one, i.e., as soon as we can certify that $\text{jsr}\,(\adjm \otimes \clAm,\adjh \otimes \clAh)<1$, we stop the algorithm. Along the iterations, the matrix set grows. The overall memory footprint needed for the matrices is the major obstacle to the scalability of the stability verification.

We implement our algorithms using the open source language Julia. 
We executed our benchmarks on an HTCondor-based high-throughput computing cluster. The different experiment sets were submitted as independent batch jobs via HTCondor, enabling opportunistic scheduling across available execute nodes. Each job was provisioned with a minimum of 16$\,$CPU cores and 20$\,$GB of RAM and was dispatched only to nodes meeting or exceeding these resource constraints. All jobs were containerized and ran Julia within the Docker image \texttt{julia:1.11.4-bookworm} (Debian Bookworm base).
Due to the resource variability of running on the HTCondor cluster, when we report running times, these are just meant to be indicative. The value that determines the scalability of the verification algorithm is the memory footprint, which we measure precisely across all our experiments.

\subsection{Proving system stability beyond reach of current methods}

To illustrate the practical impact of our compression, we consider a representative closed-loop system, namely the first one in the list above ($T=0.05$) 
given by the two mode matrices $\clAh$ (deadline hit) and $\clAm$ (deadline miss).
The target specification is the realistic weakly-hard constraint \anymiss{2}{300}.
Verifying stability directly under \anymiss{2}{300} is not feasible: the corresponding automaton is too large, and the lifted matrices in the joint spectral radius procedure quickly exceed available resources.

A standard workaround is to over-approximate \anymiss{2}{300} by reducing the window size $\weaklyhardk$ to a tractable value (e.g., \anymiss{2}{k'} with $k' \ll 300$).
This keeps the automaton small, but it also admits many additional miss patterns, which typically makes the stability verification significantly harder (the upper bound on the joint spectral radius decreases slowly and memory consumption increases quickly).
We compare this classical window-reduction approach with our compression-based over-approximation.

We run the stability analysis with \trim{2}{300}{260}. The resulting automaton has 635 states. After 19 minutes and 60 iterations of the branch-and-bound procedure, the algorithm certifies stability ($\text{jsr}<1$), with a final space consumption of 15$\,$255$\,$240 stored matrix entries.
For comparison, we select window-reduced \anymiss{2}{k'} models from the literature that yield a similar number of automaton states:
\begin{enumerate}[label=(\roman*)]
\item \anymiss{2}{36} has 630 states, just 5 fewer than \trim{2}{300}{260}. After 60 iterations, the space consumption reaches 213$\,$298$\,$470 entries (more than one order of magnitude higher than \trim{2}{300}{260}), the run time is substantially longer, and the computation still does not return a stability proof: the $\text{jsr}$ upper bound remains slightly above 1. With additional computation time, the branch-and-bound algorithm reaches 62 iterations, giving a stability proof, with a total of 276$\,$723$\,$090 matrix entries.
\item \anymiss{2}{37} has 666 states, 31 more than \trim{2}{300}{260}, and therefore could in principle capture more relevant behaviors. In practice, it remains expensive: after 60 iterations, space consumption is still 198$\,$563$\,$904 matrix entries (one order of magnitude above our method), and the algorithm again does not finish with a stability proof. The computation actually completes in roughly three hours, after 62 iterations, with a total of 294$\,$596$\,$442 matrix entries.
\end{enumerate}
Finally, we note that overly aggressive compression can also introduce enough additional behaviors to increase verification cost. For example, \trim{2}{300}{298} (stronger compression, 597 states) still finishes within 60 iterations (precisely at iteration 59), but with a much larger space consumption of 143$\,$506$\,$860 entries. This supports the intended interpretation of the compression parameter: increasing compression reduces the automaton size, and simultaneously admit extra failure patterns.
We discuss this further in Section~\ref{sec:experiments:comprate}.

These results indicate that, for essentially the same automaton size, our approach yields a much tighter over-approximation for stability verification. It retains the execution fragments that matter for the $\text{jsr}$ computation while avoiding the large number of spurious miss patterns introduced by window reduction. In other words, the classical approximations based on lowering the window for the weakly-hard any miss constraint make the language overly permissive in ways that are worse for the stability verification.

\begin{remark}
Remark~\ref{rem:26} provides an interesting quantitative perspective on our findings, if applying to the matrices we considered. We have that
\begin{equation*}
\begin{array}{rcl}
\cardinality{\langlength{\lang}\,(\anymiss{2}{36})} &\approx& \phantom{00}7.053 \cdot 1.151^{\ell} \\
\cardinality{ \langlength{\lang}\,(\anymiss{2}{37})} &\approx& \phantom{00}7.240 \cdot 1.148^{\ell} \\
\cardinality{ \langlength{\lang}\,(\anymiss{2}{300})} &\approx& \phantom{0}69.738 \cdot 1.028^{\ell} \\
\cardinality{ \langlength{\lang}\,(\trim{2}{300}{260})} &\approx& 109.224 \cdot 1.030^{\ell} \\
\end{array}
\end{equation*}
Note that the cardinality of $\trim{2}{300}{260}$ 
grows much slower than those of its competitors of similar size and is very close to the cardinality of the original constraint \anymiss{2}{300}.  Figure~\ref{fig:languagecardinality} shows how the languages grow with $\wordlength$, hinting at the root cause of the superior effectiveness of the approach we present.
\end{remark}

\begin{figure}[ht]
\centering
\begin{tikzpicture}
\begin{axis}[width=0.97\textwidth, height=5.5cm,
xmin=1, ymin=1, xmax=10000, domain=0:10000,
xlabel={$\wordlength$}, ylabel={$\cardinality{\langlength{\lang}(\cdot)}$}, grid=major, grid style={black!30,densely dashed},
legend style={draw=none,at={(0.5,1)},anchor=south},
legend columns = 2,
ymode=log,
scaled ticks=false,
tick label style={
    /pgf/number format/.cd,
    use comma=false,
    1000 sep={},
  },
]
\addplot[ultra thick, blue, no marks] {7.052704882849246*pow(1.1509294191604074,x)}; 
\addlegendentry{\anymiss{2}{36}}
\addplot[ultra thick, blue, densely dashed, no marks] {7.240425309601088*pow(1.147812338510838,x)}; 
\addlegendentry{\anymiss{2}{37}}
\addplot[ultra thick, densely dotted, blue, no marks]  {69.73800659503162*pow(1.0283046665199689,x)}; 
\addlegendentry{\anymiss{2}{300}}

\addplot[ultra thick, solid, no marks] {109.2244702586909*pow(1.03075877855802,x)};
\addlegendentry{\trim{2}{300}{260}}
\end{axis}
\end{tikzpicture}
\caption{Illustration of the language cardinality growths for increasing word length.}
\label{fig:languagecardinality}
\end{figure}
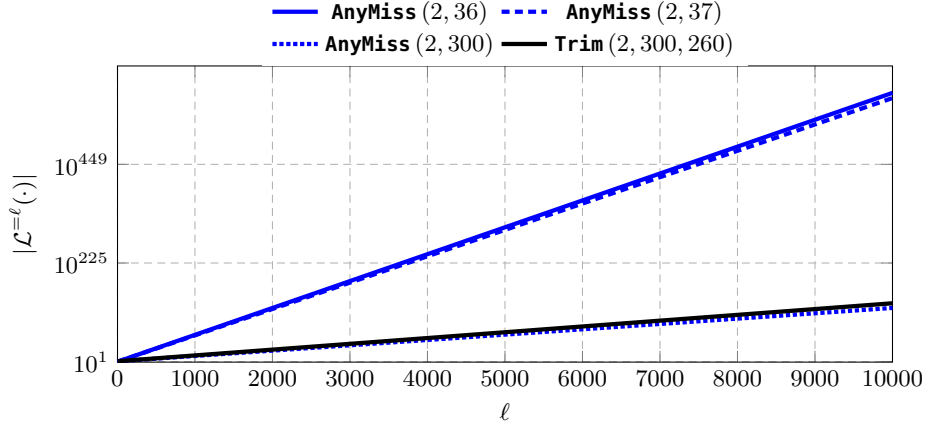

\subsection{Comparing $\aha$ and $\compaha$} 

\begin{table}[t]
\centering
\caption{Comparison between verification under \anymiss{2}{50} and \trim{2}{50}{45}.}
\label{tab:experimental_results}
\scriptsize
\def\arraystretch{1.25}%
\setlength\tabcolsep{3mm}
\begin{tabular}{@{}r|c|ccr|ccr@{}}
\hline\hline
   &                            & \multicolumn{3}{c|}{\anymiss{2}{50}~~~~~$\cardinality{\ahav} = 1225$} & \multicolumn{3}{c}{\trim{2}{50}{45}~~~~~$\cardinality{\compahav} = 100$} \\
\# & $\text{sr}\,(\clAh)$ & iter. & space & time & iter. & space & time \\
\hline\hline
 1 & 0.8346 & 57 & \phantom{0}210471660000 &   2h 20m & 42 & \phantom{0}3630960000 & 13m \\
 2 & 0.7713 & 36 & \phantom{00}38031840000 &      10m & 29 & \phantom{0}1103760000 & 2m \\
 3 & 0.7346 & 35 & \phantom{00}36033007500 &       9m & 27 & \phantom{00}875160000 & 2m \\ 
 4 & 0.7870 & 42 & \phantom{00}51159307500 &      19m & 31 & \phantom{0}1369800000 & 3m \\
 5 & 0.8491 & 63 & \phantom{0}587008485000 &   2h 50m & 46 & \phantom{0}4849200000 & 20m \\
 6 & 0.8374 & 58 & \phantom{0}250448310000 &   2h 29m & 42 & \phantom{0}3630960000 & 15m \\
 7 & 0.7498 & 36 & \phantom{00}38031840000 &      11m & 27 & \phantom{00}875160000 & 2m \\
 8 & 0.6881 & 29 & \phantom{00}25174485000 &       5m & 23 & \phantom{00}518760000 & 1m \\
 9 & 0.7084 & 28 & \phantom{00}23553810000 &       4m & 22 & \phantom{00}448560000 & 1m \\
10 & 0.7309 & 30 & \phantom{00}26849182500 &       6m & 24 & \phantom{00}596160000 & 1m \\
11 & 0.7999 & 43 & \phantom{00}53536297500 &      24m & 29 & \phantom{0}1103760000 & 2m \\
12 & 0.8197 & 50 & \phantom{00}74334960000 &      41m & 36 & \phantom{0}2217600000 & 6m \\
13 & 0.8317 & 56 & \phantom{0}177085755000 &   1h 11m & 38 & \phantom{0}2637360000 & 9m \\
14 & 0.8401 & 59 & \phantom{0}297988110000 &   1h 39m & 40 & \phantom{0}3107520000 & 10m \\
15 & 0.8558 & 66 & \phantom{0}946906380000 &    5h 6m & 50 & \phantom{0}6382800000 & 43m \\
16 & 0.8658 & 70 &           1706354685000 &   9h 46m & 54 & \phantom{0}8953920000 & 51m \\
17 & 0.8653 & 74 &           2910516210000 &  16h 54m & 59 &           16288920000 & 1h 13m \\
18 & 0.8705 & 80 &           5913626985000 &  32h 12m & 59 &           16288920000 & 1h 9m \\
19 & 0.8684 & 79 &           5291287785000 &  27h 32m & 63 &           31622040000 & 2h 6m \\
20 & 0.8670 & 83 &           8132331060000 &  49h 55m & 64 &           37888920000 & 2h 14m \\
\hline\hline
\end{tabular}
\end{table}

We now focus on \textbf{RQ2}, asking if our method is more efficient than using the classic weakly-hard framework, and if so by what magnitude.

We choose the constraint $\anymiss{2}{50}$ and its compression \trim{2}{50}{45}. 
The reason to choose in particular this constraint and this trim is that we expect that such a small number of deadline misses implies that the closed-loop system is stable, and a window $\weaklyhardk=50$ allows us to still find upper bounds on the joint spectral radius with the technique presented in~\cite{Vreman22b}, even though this may require a large amount of time.
The compression factor $\comprate=45$ is one of the possible choices, and \textbf{RQ3}, studied below, is devoted to the choice of $\comprate$. The adjacency matrices of \anymiss{2}{50} and \trim{2}{50}{45} have respectively $1225$ and $100$ states.

Table~\ref{tab:experimental_results} shows the spectral radius of the closed-loop system without deadline misses, the number of iterations that are needed to bring the upper bound from the starting value to a value that is below one, the space complexity (i.e., the number of matrix entries involved in the last iteration of the branch-and-bound algorithm), and the time it takes to get an answer for each of the benchmarks.
Unsurprisingly, if the closed-loop system is intrinsically more robust (lower $\text{sr}\,(\clAh)$), it is also easier to verify stability under sporadic deadline misses.
As can be seen \trim{\weaklyhardm}{\weaklyhardk}{\comprate} is clearly very effective in significantly reducing the computation time as well as the memory footprint needed to get to an answer. 

\subsection{The relation between $\comprate$ and efficiency}
\label{sec:experiments:comprate}

We now turn to studying \textbf{RQ3} how the compression factor $\comprate$ trades off automaton size against the practical resource usage of the stability-verification workflow. To isolate the effect of $\comprate$, we fix the plant and controller instance to system \#8 in Table~\ref{tab:experimental_results} and keep the weakly-hard parameters as $m = 2, k = 300$, varying only $\comprate$ in \trim{2}{300}{\comprate}. For each value of $\comprate$, we run the verification procedure using the branch-and-bound algorithm and record the memory usage at the last iteration, i.e., the number of matrix entries stored in the final iteration that produces the stability verdict. The resulting measurements are shown in Figure~\ref{fig:spacecompression}.

\begin{figure}[ht]
\centering
\begin{tikzpicture}
\begin{axis}[width=0.97\textwidth, height=5.5cm, ymin=4, ymax=20,
ytick={5,7,9,11,13,15,17}, legend pos=north west,
xlabel={$\comprate$}, ylabel={space}, enlargelimits=false, grid=major, grid style={black!30,densely dashed}]
\addplot[ultra thick, no marks] table[col sep=comma, x=count, y=spacediv] {data/comprate.csv};
\addlegendentry{number of matrix entries / 1$\,$000$\,$000$\,$000}
\end{axis}
\end{tikzpicture}
\caption{Illustration of the space needed for \trim{2}{300}{\comprate} when $\comprate \in [100,300]$ with system \#8 from Table~\ref{tab:experimental_results}.}
\label{fig:spacecompression}
\end{figure}
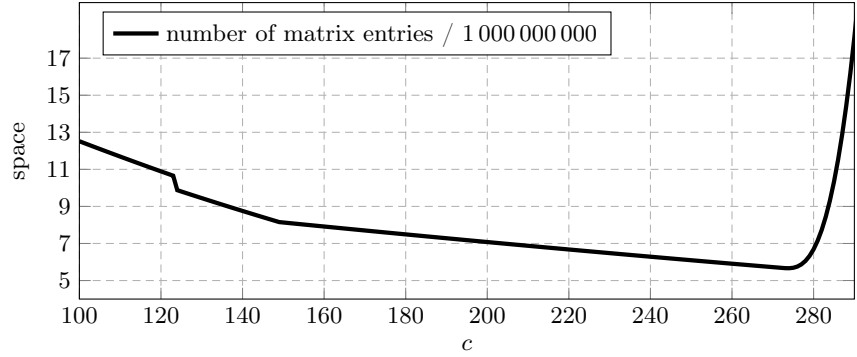

Increasing $\comprate$ reduces the size of the automaton (see Figure~\ref{fig:cardinality}), and hence shrinks the needed matrices and the per-iteration footprint: for moderate-to-large compression factors, the memory footprint decreases as $\comprate$ increases, consistent with the fact that higher $\comprate$ yields a smaller automaton. However, Figure~\ref{fig:spacecompression} also shows that this monotonicity intuition breaks at the high end: with very aggressive compression (near the upper end of the possible range), the overall memory footprint is increased massively. 
The reason is that a drastic compression factor makes the automaton structure be considerably more permissive to spurious patterns, that are costly to verify.
Notably, the original system \#8 is very robust, since the spectral radius of its closed-loop system without deadline misses $\text{sr}\,(\clAh)$ is the furthest below one. Therefore even a very imprecise over-approximation of $\anymiss{2}{300}$ can easily be certified as asymptotically stable, but the joint-spectral-radius engine is forced to retain a considerably larger set of intermediate matrices before it can certify that the joint spectral radius is below~1, thereby inflating the total memory footprint, despite the smaller automaton. For systems that are originally closer to instability ($\text{sr}\,(\clAh)$ closer to 1) the system might become unverifiable, just because it is not stable given the relaxed constraint. This effect is similar to the behaviour we observed for \anymiss{\weaklyhardm}{\weaklyhardk'} for $\weaklyhardk'\ll \weaklyhardk$ when studying \textbf{RQ1}.

\vspace{-2mm}
\section{Conclusion}
\label{sec:conclusion}

This paper has tackled the main scalability bottleneck in verifying control-system stability under realistic weakly-hard deadline-miss patterns: the combinatorial explosion of minimal constraint automata at industrially-relevant window sizes. We have proposed a sound over-approximation that replaces the original automaton with a compressed acceptor, \trim{\weaklyhardm}{\weaklyhardk}{\comprate}, in which the structured merging of critical states is controlled by the compression factor $\comprate$. We have proven that the compressed automaton simulates the original, guaranteeing language inclusion and thus preserving safety soundness. We have characterized how compression affects automaton size and the language accepted, making the tractability-conservatism trade-off explicit.
An extensive experimental campaign has demonstrated the advance over the state-of-the-art. 

There is a peculiar benefit in the method we propose. For a given verification target $\anymiss{m}{k}$, it is advisable to sweep the parameter space of $\trim{m}{k}{c}$ either analytically, based on Remark~\ref{rem:26}, or empirically, as we have done in Section~\ref{sec:experiments:comprate} to arrive at the best performing over-approximation in a controlled setting, then to be used to verify stability of challenging controller designs. The foreseeable gain over the current state-of-the-practice should be at least an order of magnitude in memory footprint and time.

\bigskip 

\noindent\textbf{Acknowledgements.}
This work was partially funded by DFG grant 389792660 as part of \href{https://perspicuous-computing.science}{TRR~248 -- CPEC} and by the EU INTERREG North Sea project STORM\_SAFE of the European Regional Development Fund.

\section*{Data-Availability Statement}
The code for the experiments presented in this paper is available at \url{https://zenodo.org/records/19710440}. The SHA265 checksum of the artifact is \texttt{d9cbfb9283cb529c13893d7ad5cbd0a99cf679ab9e59de25ff006f729ea601be}.

\bibliographystyle{plain}
\bibliography{main}

\appendix

\section{The structure of $\compaha$}
\label{sec:proofs}

\subsection*{Proof of Theorem~\ref{f bijective}}
\label{proof f bijective}

\begin{proof}
First, we show that the function is well-defined. This means that $\forall \ahanode_{k:1} = \ahanode \in \ahav, \mapfunc(\ahanode_{k:1}, \weaklyhardm, \weaklyhardk) \in \cahav$. Given the constraint \anymiss{\weaklyhardm}{\weaklyhardk} and its minimal automaton (including $\digitx$ where states are merged), for each $\ahanode_{k:1}$ there exist $\weaklyhardm$ different indices $(\indextuple_1, \dots, \indextuple_{\weaklyhardm})$ for which $\ahanode_p = 0 | \digitx$. Without loss of generality, assume that the indices are sorted in ascending order, i.e., that $i < j \implies \indextuple_i < \indextuple_j$. Let $\mapfunc(\ahanode_{k:1}, \weaklyhardm, \weaklyhardk) = \langle \cahanode_1, \dots, \cahanode_{\weaklyhardm} \rangle$.
\begin{enumerate} [label = \emph{(\roman*)}]
\item
Since the indices in the tuple are strictly ascending, the tuple with the highest possible indices is $(\weaklyhardk-\weaklyhardm + 1, \weaklyhardk-\weaklyhardm + 2, \dots, \weaklyhardk)$. Therefore, the $i$-th element of the tuple is always less than or equal to $\weaklyhardk-\weaklyhardm + i$, i.e., $\indexfunc(\ahanode_{k:1}, i) \leq \weaklyhardk-\weaklyhardm + i$. Then
\begin{align} \label{f well-defined i}
\cahanode_i = \weaklyhardk-\weaklyhardm + i - \indexfunc(\ahanode_{k:1}, i) \geq 0.
\end{align}

\item
The tuple with the smallest possible indices is $(1, 2, \dots, \weaklyhardm)$. Thus, the $i$-th element of the tuple is always larger than or equal to $i$, i.e., $\indexfunc(\ahanode_{k:1}, i) \geq i$, and therefore
\begin{align} \label{f well-defined ii}
\cahanode_i = \weaklyhardk-\weaklyhardm + i - \indexfunc(\ahanode_{k:1}, i) \leq \weaklyhardk-\weaklyhardm.
\end{align}

\item
If we consider the $i$-th and $j$-th element of the tuple, and assuming without loss of generality that $i < j$, the difference between these elements has to be at least $j - i$, i.e., $\indexfunc(\ahanode_{k:1}, j) - \indexfunc(\ahanode_{k:1}, i) \geq j - i$. Hence,
\begin{align} \label{f well-defined iii}
\cahanode_i = \weaklyhardk-\weaklyhardm + i - \indexfunc(\ahanode_{k:1}, i)  \geq \weaklyhardk-\weaklyhardm + j - \indexfunc(\ahanode_{k:1}, j) = \cahanode_j.
\end{align}

\end{enumerate}

We then show that the function is injective.
Let $\ahanode_{k:1}, \ahanode_{k:1}' \in \ahav$ and $\ahanode_{k:1} \neq \ahanode_{k:1}'$. Since $\aha = (\ahav, \ahae)$ is minimal, $\ahanode_{k:1} \not\equiv \ahanode_{k:1}'$, and there exists an index $p \in \{ 1, \dots, \weaklyhardk \}$ such that $\ahanode_p \neq \ahanode'_p$. Without loss of generality, assume that $\ahanode_p = 0$ and $\ahanode'_p = 1$. This implies that $\exists i \in \{ 1, \dots, \weaklyhardm \} \,\mid\, \indexfunc(\ahanode_{k:1}, i) = p$, and $\nexists j \in \{ 1, \dots, \weaklyhardm \} \,\mid\, \indexfunc(\ahanode_{k:1}', j) = p$.
Thus, $\indexfunc(\ahanode_{k:1}, i) \neq \indexfunc(\ahanode_{k:1}', i)$ and hence $\cahanode_i = \weaklyhardk-\weaklyhardm + i - \indexfunc(\ahanode, i) \neq \weaklyhardk-\weaklyhardm + i - \indexfunc(\ahanode', i) = \cahanode'_i$ and $\mapfunc(\ahanode) \neq \mapfunc(\ahanode')$.

Finally, we can show that the function is bijective.
The cardinalities of the domain and codomain of $\mapfunc$ are equal. $\ahav$ contains $\nicefrac{\weaklyhardk!}{(\weaklyhardk-\weaklyhardm)!\,(\weaklyhardm)!}$ elements~\cite{Vreman22}, and in $\cahav$ each element corresponds to an $\weaklyhardm$-combination with repetition of a set of size $\weaklyhardk-\weaklyhardm + 1$, which leads to the same cardinality. The bijectivity of $\mapfunc$ follows. 
\end{proof}

\subsection*{Proof of Lemma~\ref{F lambda well defined}}
\label{proof F lambda well defined}

\begin{proof}
Let $\cahaedge = (\cahanode, \alphabetletter, \cahanode') \in \cahae$ with $\cahanode = \langle \cahanode_1, \dots, \cahanode_{\weaklyhardm} \rangle$ and $\cahanode' = \langle \cahanode'_1, \dots, \cahanode'_{\weaklyhardm} \rangle$.

First, let $\cahaedge$ be in the first set of the union, then $\cahanode_i' = \max(\cahanode_i - 1, 0)$ and $\alphabetletter = \deadlinehit$. Obviously, $\cahanode \in \cahav$ and $\alphabetletter \in \alphabet$. Since $\cahanode \in \cahav$, we have $0 \leq \cahanode_i \leq \weaklyhardk-\weaklyhardm$. This implies $0 \leq \max(\cahanode_i - 1, 0) = \cahanode_i' \leq \weaklyhardk-\weaklyhardm$. Furthermore, because of $\cahanode_i \geq \cahanode_j$ for $i < j$, it follows that $\cahanode_i' = \max(\cahanode_i - 1, 0) \geq \max(\cahanode_j - 1, 0) = \cahanode_j'$ for $i < j$.

Next, let $\cahaedge$ be in the second set of the union, then $\cahanode' = \langle \weaklyhardk-\weaklyhardm, \cahanode_1, \dots, \cahanode_{\weaklyhardm - 1} \rangle$ and $\alphabetletter = \deadlinemiss$. Again, $\cahanode \in \cahav$ and $\alphabetletter \in \alphabet$ follow directly. For the first element of $\cahanode'$, $0 \leq \cahanode'_1 \leq \weaklyhardk-\weaklyhardm$ is valid because $\cahanode'_1 = \weaklyhardk-\weaklyhardm$. For the other elements, this holds because of $\cahanode'_i = \cahanode_{i - 1}$ and $0 \leq \cahanode_i \leq \weaklyhardk-\weaklyhardm$ since $\cahanode \in \cahav$. In addition, for the first two elements of $\cahanode'$, we have $\cahanode'_1 = \weaklyhardk-\weaklyhardm \geq \cahanode_1 = \cahanode'_2$. For the other elements, $\cahanode'_i \geq \cahanode'_j$ for $i < j$ applies following the same argument as above, $\cahanode'_i = \cahanode_{i - 1}$ and $\cahanode \in \cahav$. Therefore, in both cases, we have $\cahanode' \in \cahav$.
\end{proof}

\subsection*{Proof of Theorem~\ref{G and H isomorph}}
\label{isomorphism proof}

\begin{proof}
We show that there exists a bijective function $\bijfunc \! : \ahav \rightarrow \cahav$ such that
\begin{enumerate} [label = \emph{(\roman*)}]
\item{$\bijfunc(\ahainitnode) = \cahainitnode$}
\item{$\forall \, \ahanode, \ahanode' \in \ahav$ and $\alphabetletter \in \alphabet \! : (\ahanode, \alphabetletter, \ahanode') \in \ahae$ if and only if $(\bijfunc(\ahanode), \alphabetletter, \bijfunc(\ahanode')) \in \cahae$}
\item{$\ahanode$ is accepting if and only if $\bijfunc(\ahanode)$ is accepting}
\end{enumerate}
In Theorem \ref{f bijective} we already showed that the node mapping function with a fixed value for $\weaklyhardm$ and $\weaklyhardk$, $\mapfunc: \ahav \rightarrow \cahav$, is a bijective function.
\begin{enumerate} [label = \emph{(\roman*)}]
\item{
The application of the function $\indexfunc$ from Definition \ref{def indexfunc} to the initial node yields $\indexfunc(\ahanode_{init, \weaklyhardk:1}, i) = \indexfunc({\digitx^{\weaklyhardm} \deadlinehit^{\weaklyhardk - \weaklyhardm}}_{\weaklyhardk:1}, i) = \weaklyhardk - \weaklyhardm + i$. Applying $\mapfunc$ leads to $\cahanode_i = \weaklyhardk - \weaklyhardm + i - \indexfunc(\ahainitnodearg{\weaklyhardk:1}, i) = 0$, hence $\mapfunc(\ahainitnode) = \langle 0, \dots, 0 \rangle = \cahainitnode$.
}
\item{
Let $\ahanode = \ahanode_{\weaklyhardk:1} \in \ahav$ and $\mapfunc(\ahanode) = \langle \cahanode_1, \dots, \cahanode_{\weaklyhardm} \rangle = \cahanode$. Further, $\ahanode' = \ahanode'_{\weaklyhardk:1} = \ahanode_{\weaklyhardk - 1:1} \alphabetletter$ and $\mapfunc(\ahanode') = \langle \cahanode'_1, \dots, \cahanode'_{\weaklyhardm} \rangle = \cahanode'$.

First, assume that $\alphabetletter = \deadlinemiss$ and $(\ahanode, \deadlinemiss, \ahanode') \in \ahae$. Since $\ahanode$ concatenated by a zero is still in $\ahav$, $\ahanode_{\weaklyhardk - 1 : 1}$ has at most $\weaklyhardm - 1$ zeros and $\ahanode_\weaklyhardk = \digitx$. Therefore, $\indexfunc(\ahanode, \weaklyhardm) = \weaklyhardk$ and we obtain
\begin{align*}
\cahanode_{\weaklyhardm} = \weaklyhardk - \weaklyhardm + \weaklyhardm - \weaklyhardk = 0.
\end{align*}
Now, consider the values of $\indexfunc(\ahanode', i)$. The zero which is concatenated to $\ahanode$ implies $\indexfunc(\ahanode', 1) = 1$. Further, all existing zeros in $\ahanode$ are moved up by one position. Therefore, the index of the $i$-th last zero in $\ahanode'$, $i > 1$, is the index of the $(i - 1)$-th last zero in $\ahanode$ increased by one, i.e., $\indexfunc(\ahanode', i) = \indexfunc(\ahanode, i - 1) + 1$. We obtain
\begin{align*}
&\cahanode'_1 = \weaklyhardk - \weaklyhardm + 1 - \indexfunc(\ahanode', 1) = \weaklyhardk - \weaklyhardm \\
&\cahanode'_i = \weaklyhardk - \weaklyhardm + i - \indexfunc(\ahanode', i)
= \weaklyhardk - \weaklyhardm + i - (\indexfunc(\ahanode, i - 1) + 1)
= \cahanode_{i - 1} \quad \text{for } i > 0.
\end{align*}
It follows that $(\cahanode, \deadlinemiss, \cahanode') \in \cahae$.

Next, assume that $\alphabetletter = \deadlinemiss$ and there exists no node $\ahanode'' \in \ahav$ such that $(\ahanode, 0, \ahanode'') \in \ahae$. Hence, $\ahanode$ concatenated by a zero is not contained in $\ahav$. This implies that $\ahanode_{\weaklyhardk - 1:1}$ already contains $\weaklyhardm$ zeros and $\ahanode_{\weaklyhardk} \ne \deadlinemiss$. Therefore, 
$\indexfunc(\ahanode, i) < \weaklyhardk$, especially for $i = \weaklyhardm$, and we obtain
\begin{align*}
\cahanode_{\weaklyhardm} 
= \weaklyhardk - \weaklyhardm + \weaklyhardm - \indexfunc(\ahanode, \weaklyhardm)
> \weaklyhardk - \weaklyhardk
= 0.
\end{align*}
Hence, there does not exist a node $\cahanode'' \in \cahav$ such that $(\cahanode, \deadlinemiss, \cahanode'') \in \cahae$.

Finally, let $\alphabetletter = \deadlinehit$, then $(\ahanode, \deadlinehit, \ahanode') \in \ahae$. Since $\ahanode'$ is obtained by concatenating a one to $\ahanode$, the zeros in $\ahanode$ are moved up by one position. The index of the $i$-th last zero in $\ahanode'$ is the index of the $i$-th last zero in $\ahanode$ increased by one. However, increasing the indices is only possible if they have not reached their maximum value yet, i.e., $\ahanode$ has no leading zeros. The leading zeros have to be kept to ensure that $\ahanode'$ still consists of $\weaklyhardk-\weaklyhardm$ ones and $\weaklyhardm$ zeros. Therefore, the index of the $\weaklyhardm$-th zero must not be larger than $\weaklyhardk$, the index of the $(\weaklyhardm - 1)$-th zero not larger than $\weaklyhardk - 1$ etc., i.e., $\indexfunc(\ahanode', i) = \min( \indexfunc(\ahanode, i) + 1, \weaklyhardk-\weaklyhardm + i )$. We obtain
\begin{align*}
\cahanode'_i &= \weaklyhardk - \weaklyhardm + i - \indexfunc(\ahanode', i) = \weaklyhardk - \weaklyhardm + i - \min( \indexfunc(\ahanode, i) + 1, \weaklyhardk - \weaklyhardm + i ) \\
&= \weaklyhardk - \weaklyhardm + i + \max( -\indexfunc(\ahanode, i) - 1, - \weaklyhardk + \weaklyhardm - i ) \\
&= \max(\weaklyhardk - \weaklyhardm + i - \indexfunc(\ahanode, i) - 1, \weaklyhardk - \weaklyhardm + i - \weaklyhardk + \weaklyhardm - i) = \max(\cahanode_i - 1, 0).
\end{align*}
It follows that $(\cahanode, \deadlinehit, \cahanode') \in \cahae$.
}
\item{Since all nodes $\ahanode \in \ahav$ as well as all nodes $\cahanode \in \cahav$ are accepting, this condition is satisfied.}
\end{enumerate}
\end{proof}

\subsection*{Proof of Lemma~\ref{F lambda m well defined}}
\label{proof F lambda m well defined}

\begin{proof}
Let $\cahaedge = (\cahanode, \alphabetletter, \cahanode') \in \compahae$ with $\cahanode = \langle \cahanode_1, \dots, \cahanode_{\weaklyhardm} \rangle$ and $\cahanode' = \langle \cahanode'_1, \dots, \cahanode'_{\weaklyhardm} \rangle$.

First, let $\cahaedge$ be in the first set of the union, then $\cahanode_i' = \max(\cahanode_i - 1, 0)$ and $\alphabetletter = \deadlinehit$. Obviously, $\cahanode \in \compahav$ and $\alphabetletter \in \alphabet$. It remains to show that $\cahanode' \in \compahav$.
\begin{enumerate} [label = \emph{(\roman*)}]
\item{
Since $\cahanode \in \compahav$, $0 \leq \cahanode_i \leq \weaklyhardk-\weaklyhardm$ holds. This implies $0 \leq \max(\cahanode_i - 1, 0) = \cahanode_i' \leq \weaklyhardk-\weaklyhardm$.
}
\item{
Further, $\cahanode_i \geq \cahanode_j$ for $i < j$. This leads to $\cahanode_i' = \max(\cahanode_i - 1, 0) \geq \max(\cahanode_j - 1, 0) = \cahanode_j'$ for $i < j$.
}
\item{
Finally, at least one of the conditions $(\cahanode_1 - \cahanode_2) \text{ mod } \comprate = 0$ and $\cahanode_{\weaklyhardm} = 0$ is valid for $\cahanode$. If $\cahanode_{\weaklyhardm} = 0$ holds, we have
\begin{align*}
\cahanode'_{\weaklyhardm} = \max(\cahanode_{\weaklyhardm} - 1, 0) = \max(-1, 0) = 0.
\end{align*}
If $(\cahanode_1 - \cahanode_2) \text{ mod } \comprate = 0$, we can assume that $\cahanode_{\weaklyhardm} > 0$ (otherwise consider the first case). Since $\cahanode_1 \geq \cahanode_2 \geq \cahanode_{\weaklyhardm}$, also $\cahanode_1 > 0$ and $\cahanode_2 > 0$, and thus $\max(\cahanode_1 - 1, 0) = \cahanode_1 - 1$ and $\max(\cahanode_2 - 1, 0) = \cahanode_2 - 1$. We obtain
\begin{align*}
(\cahanode'_1 - \cahanode'_2) \text{ mod } \comprate &= (\max(\cahanode_1 - 1, 0) - \max(\cahanode_2 - 1, 0)) \text{ mod } \comprate \\
&= (\cahanode_1 - 1 - \cahanode_2 + 1) \text{ mod } \comprate \\
& = (\cahanode_1 - \cahanode_2) \text{ mod } \comprate \\
&= 0.
\end{align*}
}
\end{enumerate}
Next, let $\cahaedge$ be in the second set of the union, then $\cahanode' = \langle \cahanode_0, \cahanode_1, \dots, \cahanode_{\weaklyhardm - 1} \rangle$ and $\alphabetletter = \deadlinemiss$. Again, $\cahanode \in \compahav$ and $\alphabetletter \in \alphabet$ follow directly, and we show that $\cahanode' \in \compahav$.
\begin{enumerate} [label = \emph{(\roman*)}]
\item{
For $\cahanode_0 = \weaklyhardk-\weaklyhardm$, the inequality $0 \leq \cahanode_0 = \cahanode'_1 \leq \weaklyhardk-\weaklyhardm$ is obviously valid. For $\cahanode_0 = \weaklyhardk-\weaklyhardm - (\weaklyhardk-\weaklyhardm - \cahanode_1) \text{ mod } \comprate$, we first show that $(\weaklyhardk-\weaklyhardm - \cahanode_1) \text{ mod } \comprate \leq \weaklyhardk-\weaklyhardm - \cahanode_1$. Note that $\comprate > 0$ and $\weaklyhardk-\weaklyhardm - \cahanode_1 \geq 0$. If $\comprate \leq \weaklyhardk-\weaklyhardm - \cahanode_1$, then $(\weaklyhardk-\weaklyhardm - \cahanode_1) \text{ mod } \comprate < \comprate$, and thus also $(\weaklyhardk-\weaklyhardm - \cahanode_1) \text{ mod } \comprate < \weaklyhardk-\weaklyhardm - \cahanode_1$. Otherwise, if $\comprate > (\weaklyhardk-\weaklyhardm - \cahanode_1)$, then $(\weaklyhardk-\weaklyhardm - \cahanode_1) \text{ mod } \comprate = \weaklyhardk-\weaklyhardm - \cahanode_1$. Using this, we continue
\begin{alignat}{3} \label{modulo leq applied}
& & \weaklyhardk-\weaklyhardm - \cahanode_1 &\geq &(\weaklyhardk-\weaklyhardm - \cahanode_1) \text{ mod } \comprate &\geq 0 \nonumber \\
&\Leftrightarrow \quad & - \, \weaklyhardk+\weaklyhardm + \cahanode_1 &\leq & - \, (\weaklyhardk-\weaklyhardm - \cahanode_1) \text{ mod } \comprate &\leq 0 \nonumber \\
&\Leftrightarrow \quad & \cahanode_1 &\leq{} &\weaklyhardk-\weaklyhardm - (\weaklyhardk-\weaklyhardm - \cahanode_1) \text{ mod } \comprate &\leq \weaklyhardk-\weaklyhardm.
\end{alignat}
Since $\cahanode_1 \geq 0$, the inequality $0 \leq \cahanode_0 = \cahanode'_1 \leq \weaklyhardk-\weaklyhardm$ is also valid in the second case. For $i > 1$, the inequality follows from $\cahanode'_i = \cahanode_{i  - 1}$ and $\cahanode \in \compahav$.
}
\item{
For $i > 1$, $\cahanode'_i \geq \cahanode'_j$ for $i < j$ follows using the same argument that $\cahanode'_i = \cahanode_{i  - 1}$ and $\cahanode \in \compahav$. For $i = 1$, if $\cahanode'_1 = \weaklyhardk-\weaklyhardm$, this is obviously true because $\weaklyhardk-\weaklyhardm \geq \cahanode'_2$. If $\cahanode'_1 = \weaklyhardk-\weaklyhardm - (\weaklyhardk-\weaklyhardm - \cahanode_1) \text{ mod } \comprate$, using Inequality \eqref{modulo leq applied} combined with $\cahanode'_2 = \cahanode_1$ leads to the desired result.
}
\item{
If $\cahanode'_{\weaklyhardm} = 0$, then obviously $\cahanode' \in \compahav$. Otherwise, if $\cahanode'_{\weaklyhardm} = \cahanode_{\weaklyhardm - 1} > 0$, we have $\cahanode'_1 = \weaklyhardk-\weaklyhardm - (\weaklyhardk-\weaklyhardm - \cahanode_1) \text{ mod } \comprate$ and obtain
\begin{align*}
(\cahanode'_1 &- \cahanode'_2) \text{ mod } \comprate = (\weaklyhardk-\weaklyhardm - (\weaklyhardk-\weaklyhardm - \cahanode_1) \text{ mod } \comprate - \cahanode_1) \text{ mod } \comprate \\
&= ((\weaklyhardk-\weaklyhardm - \cahanode_1) - (\weaklyhardk-\weaklyhardm - \cahanode_1) \text{ mod } \comprate) \text{ mod } \comprate \\
&= \Big( \Big \lfloor \frac{\weaklyhardk-\weaklyhardm - \cahanode_1}{\comprate} \Big \rfloor \, \comprate + (\weaklyhardk-\weaklyhardm - \cahanode_1) \text{ mod } \comprate - (\weaklyhardk-\weaklyhardm - \cahanode_1) \text{ mod } \comprate \Big) \! \text{ mod } \! \comprate \\
&= \Big ( \Big \lfloor \frac{\weaklyhardk-\weaklyhardm - \cahanode_1}{\comprate} \Big \rfloor \, \comprate \Big) \text{ mod } \comprate \\
&= 0.
\end{align*}
}
\end{enumerate}
\end{proof}

\section{The size of $\compaha$}
\label{sec:size}

\subsection*{Proof of Theorem~\ref{proof size}}

As discussed on the body of the paper, $
\cardinality{\compahav} = \cardinality{\critahav} + \cardinality{\noncritahav}$ and 
\begin{align} \label{card m noncrit final}
\cardinality{\noncritahav} 
= {{\weaklyhardk - 1} \choose {\weaklyhardm - 1}}.
\end{align}
Inserting the definition of $\compahav$, we obtain
\begin{align*}
\critahav = \{ \langle \cahanode_1, \dots, \cahanode_{\weaklyhardm} \rangle \mid & \,\, 0 \leq \cahanode_i \leq x \text{ and } \cahanode_i \geq \cahanode_j \text{ for } i < j \\
&\text{ and } (\cahanode_1 - \cahanode_2) \text{ mod } \comprate = 0 \text{ and } \cahanode_{\weaklyhardm} \neq 0 \} \\
= \{ \langle \cahanode_1, \dots, \cahanode_{\weaklyhardm} \rangle \mid & \,\, 0 < \cahanode_i \leq \weaklyhardk-\weaklyhardm \text{ and } \cahanode_i \geq \cahanode_j \text{ for } i < j \\
&\text{ and } (\cahanode_1 - \cahanode_2) \text{ mod } \comprate = 0 \}.
\end{align*}
The equality of both sets can be seen as follows. The conditions $\cahanode_{\weaklyhardm} \neq 0$ and $\cahanode_{\weaklyhardm} \geq 0$ are equivalent to $\cahanode_{\weaklyhardm} > 0$, thus the last element of the tuple is strictly positive. Furthermore, since the elements of the tuple are in descending order, all elements of the tuple are strictly positive, leading to $\cahanode_i > 0$.

We split this in separate cases:
\begin{enumerate}[label=\emph{(\roman*)}]
\item When $\weaklyhardm = 1$
\begin{equation} \label{card m crit 1 final}
\critahav = \{ \langle \cahanode_1 \rangle \mid 0 < \cahanode_1 \leq \weaklyhardk-\weaklyhardm \} \implies \cardinality{ \critahav } = \weaklyhardk-\weaklyhardm.
\end{equation}
\item When $\weaklyhardm = 2$
\begin{equation}
\critahav = \{ \langle \cahanode_1, \cahanode_2 \rangle \mid 0 < \cahanode_i \leq \weaklyhardk-\weaklyhardm \land \cahanode_1 \geq \cahanode_2 \land (\cahanode_1 - \cahanode_2) \text{ mod } \comprate = 0 \}
\end{equation}
Here, $\cahanode_1$ can assume values between $1$ and $\weaklyhardk-\weaklyhardm$.
Once the value of $\cahanode_1$ is fixed, $\cahanode_2$ can assume values between $1$ and $\cahanode_1$.
Additionally, $\cahanode_1-\cahanode_2$ must be a multiple of $\comprate$.
Figure \ref{fig:possible_values_u_2} illustrates the possible values of $\cahanode_2$, which are $\cahanode_1$, $\cahanode_1 - \comprate$, $\cahanode_1 - 2\comprate$, $\dots$, $\cahanode_1 - \lfloor {\frac{\cahanode_1 - 1}{\comprate}} \rfloor \comprate$.
This means that there are $\lfloor {\frac{\cahanode_1 - 1}{\comprate}} \rfloor + 1 = \lceil {\frac{\cahanode_1}{\comprate}} \rceil$ valid values of $\cahanode_2$.
This implies that for $\weaklyhardm = 2$,
\begin{equation} \label{card m crit 2 final}
\cardinality{ \critahav } = \sum\limits_{j = 1}^{\weaklyhardk-\weaklyhardm} \Big\lceil \frac{j}{\comprate} \Big\rceil.
\end{equation}

\begin{figure}[ht]
\centering
\begin{tikzpicture}
\draw[] (0,0) -- (0.5,0);
\draw[densely dotted] (0.5,0) -- (1.5,0);
\draw[] (1.5,0) -- (2,0);
\draw[densely dotted] (2,0) -- (2.75,0);
\draw[] (2.75,0) -- (3.75,0);
\draw[densely dotted] (3.75,0) -- (5.75,0);
\draw[->] (5.75,0) -- (9,0);
\draw[] (0, 3pt) -- (0, -3pt) node[below]{$0$};
\draw[] (2, 3pt) -- (2, -3pt) node[below]{$1$};
\draw[] (2.75, 3pt) -- (2.75, -3pt);
\draw[] (3.75, 3pt) -- (3.75, -3pt);
\draw[] (5.75, 3pt) -- (5.75, -3pt);
\draw[] (6.75, 3pt) -- (6.75, -3pt);
\draw[] (7.75, 3pt) -- (7.75, -3pt) node[below]{$\cahanode_1$};
\draw [decorate,decoration={brace,amplitude=5pt}] (2,0.2)  -- +(0.75,0) node [black,midway,above=4pt] {$<\comprate$};
\draw [decorate,decoration={brace,amplitude=5pt}] (2.75,0.2)  -- +(1,0) node [black,midway,above=4pt] {$\comprate$};
\draw [decorate,decoration={brace,amplitude=5pt}] (5.75,0.2)  -- +(1,0) node [black,midway,above=4pt] {$\comprate$};
\draw [decorate,decoration={brace,amplitude=5pt}] (6.75,0.2)  -- +(1,0) node [black,midway,above=4pt] {$\comprate$};
\draw [decorate,decoration={brace,amplitude=5pt}] (7.75,-0.6)  -- +(-5,0) node [black,midway,below=4pt] {$\left\lfloor \frac{\cahanode_1 - 1}{\comprate} \right\rfloor \cdot \comprate$};
\end{tikzpicture}
\caption{Illustration of the possible values for $\cahanode_2$ when $\weaklyhardm = 2$.}
\label{fig:possible_values_u_2}
\end{figure}
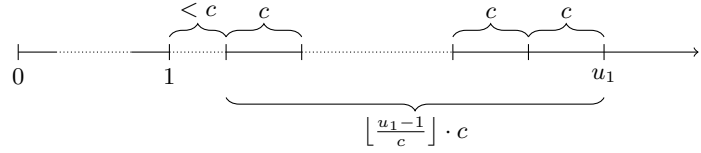
\item The last case $\weaklyhardm > 2$.
In this case we further split the set $\critahav$ into the disjoint subsets $\critdahav$
\begin{equation}
\critahav = \dot{\bigcup_{\difference \in \{ 0, \dots, \weaklyhardk - \weaklyhardm - 1 \}}} \critdahav,
\end{equation}
based on the difference $\difference$ between the second and the last element of each tuple, obtaining
\begin{equation}
\critdahav = \{ \langle \cahanode_1, \dots, \cahanode_{\weaklyhardm} \rangle \mid \cahanode \in \critahav \text{ and } \cahanode_2 - \cahanode_{\weaklyhardm} = \difference \}.
\end{equation}
We briefly justify this statement:
\begin{itemize}
\item[$\subseteq$] 
Let $\cahanode \in \critahav$. Then $\cahanode \in \critdahav$ with $\difference = \cahanode_2 - \cahanode_{\weaklyhardm}$. Since $0 < \cahanode_{\weaklyhardm} \leq \cahanode_2 \leq \weaklyhardk-\weaklyhardm$, the difference $\cahanode_2 - \cahanode_{\weaklyhardm}$ is greater than or equal to $0$ and smaller than $\weaklyhardk-\weaklyhardm$. Therefore, $\difference \in \{ 0, \dots, \weaklyhardk-\weaklyhardm - 1 \}$ and $\critdahav$ is a set of the union.
Moreover, since $\difference$ is uniquely determined by the elements $\cahanode_2$ and $\cahanode_{\weaklyhardm}$, $\cahanode$ can be contained in at most one of the sets $\critdahav$. This gives the disjointness of the sets.
\item[$\supseteq$]
Let $\cahanode$ be contained in $\critdahav$. By the definition of these sets, $\cahanode$ is contained in $\critahav$.
\end{itemize}
Since the sets $\critdahav$ are disjoint, we can calculate $\cardinality{\critahav}$ by summing up the cardinalities of the sets of the union,
\begin{equation} \label{card crit}
\cardinality{ \critahav } = \sum_{{\difference = 0}}^{\weaklyhardk-\weaklyhardm - 1} \cardinality{ \critdahav }.
\end{equation}
Inserting the definition of $\critahav$ leads to
\begin{align*}
\critdahav = \{ \langle \cahanode_1, \dots, \cahanode_{\weaklyhardm} \rangle \mid \,\, &0 < \cahanode_i \leq \weaklyhardk-\weaklyhardm \land \cahanode_i \geq \cahanode_j \text{ for } i < j \\
&\land (\cahanode_1 - \cahanode_2) \text{ mod } \comprate = 0 \land \cahanode_2 - \cahanode_{\weaklyhardm} = \difference \}.
\end{align*}
We continue by further splitting this set into disjoint subsets $\critdeahav$ based on the value of $\cahanode_2$ as follows,
\begin{equation}
\critdahav = \dot{\bigcup_{\e \in \{ \difference + 1, \dots, \weaklyhardk-\weaklyhardm \}}} \critdeahav
\end{equation}
where
\begin{equation}
\critdeahav = \{ \langle \cahanode_1, \dots, \cahanode_{\weaklyhardm} \rangle \mid \cahanode \in \critdahav \land \cahanode_2 = \e \}.
\end{equation}
We briefly show the validity of this statement.
\begin{itemize}
\item[$\subseteq$]
Let $\cahanode \in \critdahav$. Then $\cahanode \in \critdeahav$ with $\e = \cahanode_2$. It remains to show that $\cahanode_2 \in \{ \difference + 1, \dots, \weaklyhardk-\weaklyhardm \}$. Transforming $\difference = \cahanode_2 - \cahanode_{\weaklyhardm}$ leads to $\cahanode_2 = \difference + \cahanode_{\weaklyhardm}$, and since $\cahanode_{\weaklyhardm} > 0$ we have $\cahanode_2 \geq \difference + 1$, which gives the lower bound. The upper bound $\cahanode_2 \leq \weaklyhardk-\weaklyhardm$ follows directly from the definition. Moreover, since $\e$ is uniquely determined by $\cahanode_2$, $\cahanode$ can be contained in at most one of the sets $\critdeahav$. This gives the disjointness of the sets.
\item[$\supseteq$]
Let $\cahanode$ be contained in $\critdeahav$. By definition of these sets, $\cahanode$ is contained in $\critdahav$.
\end{itemize}
The disjointness of the subsets $\critdeahav$ allows us to calculate
\begin{equation} \label{card crit d}
\cardinality{ \critdahav } = \sum_{{\e = \difference + 1}}^{\weaklyhardk-\weaklyhardm} \cardinality{ \critdeahav },
\end{equation}
and inserting the definition of $\critdahav$ gives us
\begin{align*}
\critdeahav = \{ \langle \cahanode_1, \dots, \cahanode_{\weaklyhardm} \rangle \mid \,\, &0 < \cahanode_i \leq \weaklyhardk-\weaklyhardm \text{ and } \cahanode_i \geq \cahanode_j \text{ for } i < j \\
&\land (\cahanode_1 - \cahanode_2) \text{ mod } \comprate = 0 \\
&\land \cahanode_2 - \cahanode_{\weaklyhardm} = \difference \land \cahanode_2 = \e \}.
\end{align*}
To determine $\cardinality{\critdeahav}$, we first consider the number of alternatives for $\cahanode_1$, and then the number of possible values for $\cahanode_3, \dots, \cahanode_{\weaklyhardm - 1}$. In fact, $\cahanode_2$ is determined based on $\e$ and $\cahanode_{\weaklyhardm}$ is obtained from $\cahanode_2$ and $\difference$.
The value of $\cahanode_1$ must satisfy $\cahanode_2 \leq \cahanode_1 \leq \weaklyhardk-\weaklyhardm$.
In addition, $\cahanode_1-\cahanode_2$ has to be a multiple of $\comprate$.
Figure \ref{fig:possible_values_u_1} illustrates the values that $\cahanode_1$ can take, namely $\cahanode_2$, $\cahanode_2 + \comprate$, $\cahanode_2 + 2\comprate$, $\dots$, $\cahanode_2 + \lfloor \frac{\weaklyhardk-\weaklyhardm - \cahanode_2}{\comprate} \rfloor \comprate$.
This means that there are $\lfloor \frac{\weaklyhardk-\weaklyhardm - \cahanode_2}{\comprate} \rfloor + 1 = \lceil \frac{\weaklyhardk-\weaklyhardm - \cahanode_2 + 1}{\comprate} \rceil$ alternatives $\cahanode_1$.

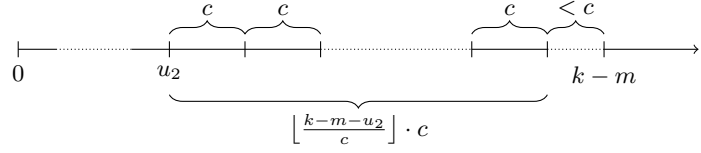
\begin{figure}[ht]
\centering
\begin{tikzpicture}
\draw[] (0,0) -- (0.5,0);
\draw[densely dotted] (0.5,0) -- (1.5,0);
\draw[] (1.5,0) -- (4,0);
\draw[densely dotted] (4,0) -- (6,0);
\draw[] (6,0) -- (7,0);
\draw[densely dotted] (7,0) -- (7.75,0);
\draw[->] (7.75,0) -- (9,0);
\draw[] (0, 3pt) -- (0, -3pt) node[below]{$0$};
\draw[] (2, 3pt) -- (2, -3pt) node[below]{$u_2$};
\draw[] (3, 3pt) -- (3, -3pt);
\draw[] (4, 3pt) -- (4, -3pt);
\draw[] (6, 3pt) -- (6, -3pt);
\draw[] (7, 3pt) -- (7, -3pt);
\draw[] (7.75, 3pt) -- (7.75, -3pt) node[below]{$\weaklyhardk-\weaklyhardm$};
\draw [decorate,decoration={brace,amplitude=5pt}] (2,0.2)  -- +(1,0) node [black,midway,above=4pt] {$\comprate$};
\draw [decorate,decoration={brace,amplitude=5pt}] (3,0.2)  -- +(1,0) node [black,midway,above=4pt] {$\comprate$};
\draw [decorate,decoration={brace,amplitude=5pt}] (6,0.2)  -- +(1,0) node [black,midway,above=4pt] {$\comprate$};
\draw [decorate,decoration={brace,amplitude=5pt}] (7,0.2)  -- +(0.75,0) node [black,midway,above=4pt] {$<\comprate$};
\draw [decorate,decoration={brace,amplitude=5pt}] (7,-0.6)  -- +(-5,0) node [black,midway,below=4pt] {$\left\lfloor \frac{\weaklyhardk-\weaklyhardm - \cahanode_2}{\comprate} \right\rfloor \cdot \comprate$};
\end{tikzpicture}
\caption{Illustration of the possible values for $\cahanode_1$ when $\weaklyhardm > 2$.}
\label{fig:possible_values_u_1}
\end{figure}

We continue by considering the number of possibilities to choose the values $\cahanode_3, \dots, \cahanode_{\weaklyhardm - 1}$. For $\cahanode$ to be contained in the set $\critdeahav$, the following must be true for these values,
\begin{equation} \label{descending order}
\cahanode_2 \geq \cahanode_3 \geq \dots \geq \cahanode_{\weaklyhardm - 1} \geq \cahanode_{\weaklyhardm}.
\end{equation}
Rewriting the condition $\difference = \cahanode_2 - \cahanode_{\weaklyhardm}$, leads to
\begin{equation}
\difference = (\cahanode_2 - \cahanode_3) + (\cahanode_3 - \cahanode_4) + \dots + (\cahanode_{\weaklyhardm - 1} - \cahanode_{\weaklyhardm}) 
= \sum_{i = 3}^{\weaklyhardm} \cahanode_{i - 1} - \cahanode_i.
\end{equation}
The requirement for the values $\cahanode_3, \dots, \cahanode_{\weaklyhardm - 1}$ stated in \eqref{descending order} is equivalent to each summand of this sum being greater than or equal to zero, i.e., $\cahanode_{i - 1} - \cahanode_i \geq 0$.
Defining $\difference_i = \cahanode_{i - 1} - \cahanode_i$, we obtain
\begin{equation}
\difference = \sum_{i = 3}^{\weaklyhardm} \difference_i \quad \text{where } \difference_i \in \natnumbers.
\end{equation}
This transforms the problem of how many possibilities there are to choose the values $\cahanode_3, \dots, \cahanode_{\weaklyhardm - 1}$ for a given $\difference$ and $\cahanode_2$, into the problem of how many possibilities there are to split the natural number $\difference$ into $\weaklyhardm - 2$ natural non-negative summands.
To solve this, we prove the following lemma.

\begin{lemma} \label{Possibilities split natural number}
Let $n \in \natnumbers$. Then there exist $n + s - 1 \choose n$ different possibilities to choose $s \in \natnumbers^{+}$ summands $n_i$ such that 
\begin{align*}
n = \sum_{i = 1}^s n_i \quad \text{and} \quad n_i \in \natnumbers.
\end{align*}
\end{lemma}

\begin{proof}
We prove this by induction over $s$ and $n$. To simplify notation, we define $N(n, s)$ as the number of possibilities to split the number $n$ into $s$ summands according to the assumptions of the lemma. First, let $s = 1$. Obviously,
\begin{align*}
N(n, 1) = 1 = {n \choose n} = {n + 1 - 1 \choose n}.
\end{align*}
Now, we assume
\begin{align} \label{IH s}
N(n, s) = {n + s - 1 \choose n} \quad \text{for an } s \in \natnumbers^{+} \quad \text{(induction hypothesis)},
\end{align}
and prove the statement for $s + 1$. The central observation is the following, splitting $n$ into two summands, there are the following $n + 1$ possibilities,
\begin{align*}
n = 0 + n
= 1 + (n - 1)
= \dots
= (n - 1) + 1
= n + 0.
\end{align*}
To achieve the splitting of $n$ into $s + 1$ summands instead of two summands, we split the second summand of each expression into $s$ summands. Thus, for the number of possibilities to split $n$ into $s + 1$ summands we obtain
\begin{align} \label{eq n s + 1}
N(n, s + 1) = N(n, s) + N(n - 1, s) + \dots + N(0, s) = \sum_{i = 0}^{n} N(i, s).
\end{align}
Next, we prove
\begin{align} \label{induction over n}
\sum_{i = 0}^{n} {i + s - 1 \choose i} = {n + s \choose n}
\end{align}
by induction over $n$. Let $n = 0$, then
\begin{align*}
\sum_{i = 0}^{0} {i + s - 1 \choose i} = { 0 + s - 1 \choose 0 } = 1 = { 0 + s \choose 0 }.
\end{align*}
We assume 
\begin{align*}
\sum_{i = 0}^{n} {i + s - 1 \choose i} = {n + s \choose n} \quad \text{for an } n \in \natnumbers \quad \text{(induction hypothesis)},
\end{align*}
and show the claim for $n + 1$,
\begin{align*}
\sum_{i = 0}^{n + 1} {i + s - 1 \choose i} 
&= {(n + 1) + s - 1 \choose n + 1} + \sum_{i = 0}^{n} {i + s - 1 \choose i} \\
&\overset{}{=} {n + s \choose n + 1} + {n + s \choose n}
\overset{}{=} {n + s + 1 \choose n + 1},
\end{align*}
which proves Equation \eqref{induction over n}.
Returning to Equation \eqref{eq n s + 1} and using the induction hypothesis from Equation \eqref{IH s}, we continue,
\begin{align*}
N(n, s + 1) 
\overset{}{=} \sum_{i = 0}^{n} N(i, s) 
\overset{}{=} \sum_{i = 0}^{n} {i + s - 1 \choose i} 
\overset{}{=} {n + s \choose n} 
= {n + (s + 1) - 1 \choose n},
\end{align*}
which finishes the induction over $s$ and proves the lemma.
\end{proof}

Returning to the original problem of how many possibilities there are to split the number $\difference$ into $\weaklyhardm - 2$ natural non-negative summands, applying Lemma~\ref{Possibilities split natural number} we obtain ${\difference + \weaklyhardm - 3 \choose \difference}$ possibilities.
Hence, there are also $\difference + \weaklyhardm - 3 \choose \difference$ possibilities to choose the values $\cahanode_3, \dots, \cahanode_{\weaklyhardm - 1}$.

Now we can combine the results.
First, the independence of the choices of the values of $\cahanode_1$ and $\cahanode_3, \dots, \cahanode_{\weaklyhardm - 1}$ implies that
\begin{equation}
\cardinality{ \critdeahav } = \left\lceil \frac{\weaklyhardk-\weaklyhardm - \e + 1}{\comprate} \right\rceil \, {\difference + \weaklyhardm - 3 \choose \difference}.
\end{equation}

Using Equation \eqref{card crit d}, we obtain
\begin{equation} \label{card crit d final}
\begin{array}{rcl}
\cardinality{ \critdahav } 
&=& \sum_{{\e = \difference + 1}}^{\weaklyhardk-\weaklyhardm} \, \cardinality{ \critdeahav } \\[2mm]
&=& \sum_{{\e = \difference + 1}}^{\weaklyhardk-\weaklyhardm} \left\lceil \frac{\weaklyhardk-\weaklyhardm - \e + 1}{\comprate} \right\rceil \, {\difference + \weaklyhardm - 3 \choose \difference} \\[2mm]
&=& {\difference + \weaklyhardm - 3 \choose \difference} \, \sum_{{\e = \difference}}^{\weaklyhardk-\weaklyhardm - 1} \left\lceil \frac{\weaklyhardk-\weaklyhardm - (\e + 1) + 1}{\comprate} \right\rceil \\[2mm]
&=& {\difference + \weaklyhardm - 3 \choose \difference} \, \sum_{{\e = \difference - \weaklyhardk+\weaklyhardm}}^{-1} \left\lceil \frac{\weaklyhardk-\weaklyhardm - (\e + \weaklyhardk-\weaklyhardm)}{\comprate} \right\rceil \\[2mm]
&=& {\difference + \weaklyhardm - 3 \choose \difference} \, \sum_{{\e = 1}}^{\weaklyhardk-\weaklyhardm - \difference} \left\lceil \frac{\e}{\comprate} \right\rceil.
\end{array}
\end{equation}

Inserting this result into Equation \eqref{card crit} yields that, for $\weaklyhardm > 2$
\begin{equation} \label{card m crit 3 final}
\cardinality{ \critahav } 
\overset{}{=} \sum_{{\difference = 0}}^{\weaklyhardk-\weaklyhardm - 1} \, \cardinality{ \critdahav }
\overset{}{=} \sum_{{\difference = 0}}^{\weaklyhardk-\weaklyhardm - 1} 
{\difference + \weaklyhardm - 3 \choose \difference} \, \sum_{{\e = 1}}^{\weaklyhardk-\weaklyhardm - \difference} \left\lceil \frac{\e}{\comprate} \right\rceil . 
\end{equation}

\end{enumerate}

Summarizing the results of Equations~\eqref{card m},~\eqref{card m noncrit final},~\eqref{card m crit 1 final},~\eqref{card m crit 2 final} and \eqref{card m crit 3 final} for the different cases based on the difference between $\weaklyhardk$ and $\weaklyhardm$, and recalling that $\weaklyhardm>0$, we obtain that the cardinality of the compressed automaton for \trim{\weaklyhardm}{\weaklyhardk}{\comprate} is
\begin{equation}
\cardinality{\compahav} = {\weaklyhardk - 1 \choose \weaklyhardm - 1} + 
\begin{cases}
   \weaklyhardk-\weaklyhardm \vphantom{\sum\limits_{j = 1}^{\weaklyhardk-\weaklyhardm}}
      & \quad\weaklyhardm = 1 \\[3mm]
   \sum\limits_{j = 1}^{\weaklyhardk-\weaklyhardm} \left\lceil \frac{j}{\comprate} \right\rceil
      & \quad\weaklyhardm = 2 \\[3mm]
   \sum\limits_{i = 0}^{\weaklyhardk-\weaklyhardm - 1} 
     {\weaklyhardm - 3 + i \choose i} \sum\limits_{j = 1}^{\weaklyhardk-\weaklyhardm - i} \left\lceil \frac{j}{\comprate} \right\rceil 
      & \quad\weaklyhardm > 2
\end{cases}.
\end{equation}
\section{The language of $\compaha$}
\label{sec:cardinality}
We discuss the language cardinality of $\compaha$. 
Let $n$ be the number of its states, i.e., $n = \cardinality{\compahav} = \cardinality{ \{ \cahanode_1, \dots, \cahanode_n \} }$, and $\adj \in \realnumbersmatrix{n}{n}$ the adjacency matrix defined by
\begin{align*}
\adjelement_{ij} = 
\begin{cases}
1 \quad & \text{if } (\cahanode_i, \deadlinehit, \cahanode_j) \in \compahae \text{ or } (\cahanode_i, \deadlinemiss, \cahanode_j) \in \compahae \\
0 & \text{otherwise.}
\end{cases}
\end{align*}
The number of accepted words of length $\wordlength$ can be calculated by $\cardinality{ \langlength{\lang}\,(\compaha) } = \initstatevector \, \adj^\wordlength \accstatevector$, where $\initstatevector \in \realnumbersmatrix{1}{n}$ is a vector representing the initial state $\cahainitnode$,
(i.e.,~$\initstatevector_i = 1$ if $u_i$ is the initial state and $0$ otherwise),
and $\accstatevector \in \realnumbersmatrix{n}{1}$ is a vector representing the accepting states, which for us means $\accstatevector =
\begin{pmatrix}
1 & 1 & \dots & 1
\end{pmatrix}^\top.$ 
Without loss of generality, we assume that $\cahanode_1 = \cahainitnode$. Then we have
$\initstatevector =
\begin{pmatrix}
1 & 0 & \dots & 0
\end{pmatrix}$,
Since $\cardinality{ \langlength{\lang}\,(\compaha) }$ is one-dimensional
, it equals its transpose,
\begin{align} \label{eq:card_as_one_norm}
\cardinality{ \langlength{\lang}\,(\compaha) }
= \accstatevector^\top (\adj^\top)^\wordlength \initstatevector^\top 
= \begin{pmatrix}
1 & 1 & \dots & 1
\end{pmatrix} (\adj^\top)^\wordlength \, \initstatevector^\top
= \onenorm{(\adj^\top)^\wordlength \, \initstatevector^\top},
\end{align}
 where the $1$-norm is defined as usual. 

Next, we establish some useful general properties of matrices for matrix $\adj$. A matrix $A = (a_{ij}) \in \realnumbersmatrix{m}{n}$ or vector $x \in \realnumbers^{n}$ is called strictly positive, if its entries are strictly positive, i.e., $a_{ij} > 0$, respectively $x_i > 0$. We write $A > 0$, respectively $x > 0$. A matrix or vector is called non-negative, if its entries are non-negative, i.e., $a_{ij} \geq 0$, respectively $x_i \geq 0$. We write $A \geq 0$, respectively $x \geq 0$. A square matrix $A = (a_{ij}) \in \realnumbersmatrix{n}{n}$ is said to be irreducible, if for each pair $i$ and $j$ of indices, there exists a sequence of indices $i = i_1, i_2, \dots, i_{l - 1}, i_l = j$ such that $a_{i_k, i_{k+1}} \neq 0$ for all $1 \leq k \leq l-1$. A square non-negative matrix $A \in \realnumbersmatrix{n}{n}$ is called primitive, if $\exists k \in \mathbb{N}^{+}$ such that $A^k > 0$. Primitivity is invariant under transpose. We continue by showing that some of these properties apply to $\adj$.

\begin{lemma}\label{adj non-negative}\label{adj irreducible}
\label{adj primitive}
$\adj$ is a non-negative, irreducible and primitive matrix.
\end{lemma}
\begin{proof}
Since $\adjelement_{ij}$ can only take the values $1$ and $0$, the non-negativity follows directly.

To prove the irreducibility of a matrix, it is sufficient to show that its directed adjacency graph is strongly connected ~\cite[Theorem 8.43]{KnabnerLinAlg}. 
This graph is given by $\caha$. From each node, after taking $\weaklyhardk$ one-transitions, the initial node $\cahainitnode$ is reached. From $\cahainitnode$ every node can be reached by an appropriate combination of one- and zero-transitions.

According to \cite[Example 8.3.3]{MeyerMatAnaAndAppLinAlg}, 
a non-negative and irreducible matrix $A = (a_{ij}) \in \realnumbersmatrix{n}{n}$ is primitive if $A$ has at least one strictly positive entry on its diagonal, i.e., there exists an $i$ such that $a_{ii} > 0$. As already shown, $\adj$ is non-negative and irreducible. In addition, there always exists a transition from the initial state to itself, thus $u_{11} = 1$, which proves the primitivity of $\adj$. The primitivity of $\adj^\top$ follows from the invariance of primitivity under transpose.
\end{proof}
Next, we restate some well-known theorems concerning primitive matrices.
\begin{lemma}[Perron-Frobenius for Primitive Matrices] \label{Perron-Frobenius}
Let $A \in \realnumbersmatrix{n}{n}$ be a primitive matrix. Then there exists a real-valued, strictly positive and simple eigenvalue that is larger than any other eigenvalue of $A$. Moreover, the associated eigenvector is strictly positive and unique up to constant multiples.
\end{lemma}
\begin{proof}
See Theorem 1.1 in~\cite{SenetaNonNegMatAndMarkCha}. 
\end{proof}

\begin{lemma}[Limits of Powers for Primitive Matrices] \label{limits of powers}
Let $A \in \realnumbersmatrix{n}{n}$ be a primitive matrix with the eigenvalue $\eigenvalue$ that equals the spectral radius of $A$. Further, let $\eigenvectorx$ and $\eigenvectory$ be the eigenvectors of $A$ and $A^\top$ associated to $\eigenvalue$. Then the following applies,
\begin{align*}
\lim_{k \rightarrow \infty} \left( \frac{A}{\eigenvalue} \right) ^k = \frac{\eigenvectorx \eigenvectory^\top}{\eigenvectory^\top \eigenvectorx} > 0.
\end{align*}
\end{lemma}
\begin{proof}
See Equation (8.3.10) in \cite{MeyerMatAnaAndAppLinAlg}. 
The normalization of $\eigenvectorx$ and $\eigenvectory$ that is assumed there is not necessary, since it holds that 
$${\eigenvectorx \eigenvectory^\top}({\eigenvectory^\top \eigenvectorx})^{-1} =
({\eigenvectorx}/{\norm{\eigenvectorx}}) ({\eigenvectory}/{\norm{\eigenvectory}})^\top
\left( {(\eigenvectory}/{\norm{\eigenvectory}})^\top ({\eigenvectorx}/{\norm{\eigenvectorx}}) \right)^{-1}.$$
\end{proof}

\begin{corollary}\label{limits of powers with vector}
Let $A \in \realnumbersmatrix{n}{n}$ be a primitive matrix with the eigenvalue $\eigenvalue$ that equals the spectral radius of $A$, and let $\eigenvectorx$ and $\eigenvectory$ be the eigenvectors of $A$ and $A^\top$ associated to $\eigenvalue$. Further, let $t \in \realnumbers^n$ be a non-negative vector that is not the null vector, $0 \neq t \geq 0$. Then the following statements apply,
\begin{enumerate} [label = \emph{(\roman*)}]
\item \label{limit (i)} $\begin{aligned}
\lim\limits_{k \rightarrow \infty} \frac{A^k t}{\eigenvalue^k} = \frac{\eigenvectory^\top t}{\eigenvectory^\top \eigenvectorx}\eigenvectorx \geq 0 \text{ and } \neq \overrightarrow{0},
\end{aligned}$
\item \label{limit (ii)} $\begin{aligned}
\lim_{k \rightarrow \infty} \frac{\norm{A^k t}}{\eigenvalue^k} = \norm{ \frac{\eigenvectory^\top t}{\eigenvectory^\top \eigenvectorx}\eigenvectorx } \quad \text{for an arbitrary norm } \norm{\quad},
\end{aligned}$
\item $\begin{aligned}
\lim_{k \rightarrow \infty} \frac{\norm{A^{k + 1}t}}{\norm{A^k t}} = \eigenvalue \quad \text{for an arbitrary norm } \norm{\quad}.
\end{aligned}$
\end{enumerate}
\end{corollary}
This is mainly a consequence of Lemma~\ref{limits of powers}.
\begin{proof}
\begin{enumerate} [label = \emph{(\roman*)}]
\item{
This follows directly from Lemma~\ref{limits of powers}, which applies since $A$ is a primitive matrix,
\begin{align*}
\lim\limits_{k \rightarrow \infty} \frac{A^k t}{\eigenvalue^k}
= \lim\limits_{k \rightarrow \infty} \left( \frac{A}{\eigenvalue}\right)^k t
\overset{}
{=} \left( \frac{\eigenvectorx \eigenvectory^\top}{\eigenvectory^\top \eigenvectorx} \right) t
\overset{}
{=} \frac{\eigenvectorx (\eigenvectory^\top t)}{\eigenvectory^\top \eigenvectorx}
\overset{}
{=} \frac{(\eigenvectory^\top t) \eigenvectorx}{\eigenvectory^\top \eigenvectorx}
= \frac{\eigenvectory^\top t}{\eigenvectory^\top \eigenvectorx} \eigenvectorx.
\end{align*}
Moreover, the term is non-negative and not the null vector, since, looking at the third term $\left( \frac{\eigenvectorx \eigenvectory^\top}{\eigenvectory^\top \eigenvectorx} \right) t$, $\frac{\eigenvectorx \eigenvectory^\top}{\eigenvectory^\top \eigenvectorx}$ is strictly positive according to Lemma~\ref{limits of powers}, and $t$ is non-negative and not the null vector by assumption of this theorem, so by the definition of matrix multiplication their product is at least non-negative and not the null vector.
}
\item{
The second statement is a consequence of the first,
\begin{align*}
\lim_{k \rightarrow \infty} \frac{\norm{A^k t}}{\eigenvalue^k}
\overset{}
{=} \lim_{k \rightarrow \infty} \frac{\norm{A^k t}}{ \abs{\eigenvalue^k}}
\overset{}
{=} \lim_{k \rightarrow \infty} \norm{ \frac{{A^k t}}{ {\eigenvalue^k}}}
\overset{}
{=} \norm {\lim_{k \rightarrow \infty} { \frac{{A^k t}}{ {\eigenvalue^k}}}}
\overset{}
{=} \norm{ \frac{\eigenvectory^\top t}{\eigenvectory^\top \eigenvectorx}\eigenvectorx }.
\end{align*}
}
\item{
The third is again a consequence of the second,
\begin{align*}
\lim_{k \rightarrow \infty} \frac{\norm{A^{k + 1}t}}{\norm{A^k t}}
&= \lim_{k \rightarrow \infty} \frac{\eigenvalue^{k + 1} \norm{A^{k + 1}t}}{\eigenvalue^{k + 1} \norm{A^k t} \quad } 
= \eigenvalue \lim_{k \rightarrow \infty} \frac{\eigenvalue^k \quad \norm{A^{k + 1}t}}{\eigenvalue^{k + 1} \norm{A^k t} \quad} \\
&= \eigenvalue \lim_{k \rightarrow \infty} \left( \left( \frac{\norm{A^k t}}{\eigenvalue^{k}}\right)^{-1} \frac{\norm{A^{k + 1}t}}{\eigenvalue^{k + 1}} \right) \\
&\overset{}{=} \eigenvalue \left( \lim_{k \rightarrow \infty} \frac{\norm{A^k t}}{\eigenvalue^{k}} \right)^{-1} \left( \lim_{k \rightarrow \infty} \frac{\norm{A^{k + 1}t}}{\eigenvalue^{k + 1}} \right) \\
&\overset{}{=} \eigenvalue \norm{ \frac{\eigenvectory^\top t}{\eigenvectory^\top \eigenvectorx}\eigenvectorx }^{-1} \norm{ \frac{\eigenvectory^\top t}{\eigenvectory^\top \eigenvectorx}\eigenvectorx }
= \eigenvalue.
\end{align*}
}
\end{enumerate}
\end{proof}
Now we return to $\adj$ and $\adj^\top$. According to Lemma~\ref{Perron-Frobenius}, there exists an eigenvalue $\eigenvalue \in \realnumbers$, $\eigenvalue > 0$, which is equal to the spectral radius of these matrices.
In addition, the associated eigenvectors $\eigenvectory$ and $\eigenvectorx$ of $\adj$ and $\adj^\top$ are strictly positive. Since $\adj^\top \in \realnumbersmatrix{n}{n}$ is a primitive matrix and $\initstatevector^\top \in \realnumbers^{n}$ is a non-negative vector that is not the null vector, we can apply Lemma~\ref{limits of powers with vector} and obtain
\begin{align*} 
\lim_{\wordlength \rightarrow \infty} \frac{\onenorm{(\adj^\top)^\wordlength \initstatevector^\top}}{\eigenvalue^\wordlength} = \onenorm{ \frac{\eigenvectory^\top \initstatevector^\top}{\eigenvectory^\top \eigenvectorx} \eigenvectorx},
\quad\text{ and }\quad
\lim_{\wordlength \rightarrow \infty} \frac{\onenorm{(\adj^\top)^{\wordlength+1} \initstatevector^\top}}{\onenorm{(\adj^\top)^\wordlength \initstatevector^\top}} = \eigenvalue.
\end{align*}

Applying Equation \eqref{eq:card_as_one_norm}, these results allow us to draw conclusions about the behavior of $\cardinality{ \langlength{\lang}\,(\compaha) }$ when $\wordlength$ approaches infinity. We have
\begin{align*}
\lim_{\wordlength \rightarrow \infty} \frac{\cardinality{ \langlength{\lang}\,(\compaha) }}{\eigenvalue^\wordlength} 
\overset{}{=} \lim_{\wordlength \rightarrow \infty} \frac{\onenorm{(\adj^\top)^\wordlength \initstatevector^\top}}{\eigenvalue^\wordlength}
\overset{}{=}
\onenorm{ \frac{\eigenvectory^\top \initstatevector^\top}{\eigenvectory^\top \eigenvectorx} \eigenvectorx}
\end{align*}
and
\begin{align*}
\lim_{\wordlength \rightarrow \infty} \frac{\cardinality{ \lang^{=\wordlength + 1} \,(\compaha) }}{\cardinality{ \langlength{\lang}\,(\compaha) }}
\overset{}{=} \lim_{\wordlength \rightarrow \infty} \frac{\onenorm{(\adj^\top)^{\wordlength+1} \initstatevector^\top}}{\onenorm{(\adj^\top)^\wordlength \initstatevector^\top}} 
\overset{}{=} \eigenvalue.
\end{align*}
Therefore, for large values of $\wordlength$, the number of accepted words of length $\wordlength$ can be approximated by an exponential function with the spectral radius of the adjacency matrix as the growth rate and a factor composed of the eigenvectors of the adjacency matrix and its transpose associated to this eigenvalue and the initial vector,
\begin{align} \label{card ev approx}
\cardinality{ \langlength{\lang}\,(\compaha) } \approx \onenorm{ \frac{\eigenvectory^\top \initstatevector^\top}{\eigenvectory^\top \eigenvectorx} \eigenvectorx} \cdot \eigenvalue^\wordlength.
\end{align}
Since  $\aha$ is isomorphic to $\compaha$ with $c=1$, the result applies to $\aha$ as well, and thus to all automata we consider. 

\end{document}